\newif\ifjournal
\journalfalse
\ifjournal
  \documentclass{mrlart4}

  \address{Department of Mathematical Sciences, University of Alaska Fairbanks, Fairbanks, Alaska 99775}
  \email{david.maxwell@uaf.edu}
\else
\documentclass{article}
\usepackage{times}
\usepackage{geometry}
\usepackage{graphicx}
\usepackage{amsmath,amsthm,amssymb}
\usepackage[lf]{MinionPro}
\usepackage{color}
\definecolor{labelkey}{rgb}{0,0,.75}
\definecolor{MyGreen}{rgb}{0,.6,.2}
\definecolor{MyDarkBlue}{rgb}{.1,.1,.75}
\usepackage[colorlinks,citecolor=MyGreen, linkcolor=MyDarkBlue]{hyperref}

\date{September 30, 2009}
\fi
\ifjournal
\title[Conformal parameterizations]{A model problem for conformal parameterizations of the Einstein constraint equations}
\else
\title{A model problem for conformal parameterizations of the Einstein constraint equations}
\fi
\author{David Maxwell}

\DeclareMathOperator{\ck}{\bf L}
\DeclareMathOperator{\Lap}{\Delta}

\DeclareMathOperator{\tr}{\rm tr}
\DeclareMathOperator{\Lie}{\mathcal{L}}
\newcommand{\dexp}{d^{-\frac{2q}{n}}\;}

\newcommand{\dV}{\; dV}
\newcommand{\abs}[1]{\left|#1\right|}

\renewcommand{\div}{\mathop{\rm div}\nolimits}
\newcommand{\yamabe}{{\mathcal Y}}

\newcommand{\ra}{\rightarrow}

\newcommand{\Reals}{\mathbb{R}}
\newcommand{\Nats}{\mathbb{R}}
\newcommand{\be}{\begin{equation}}
\newcommand{\bel}[1]{\begin{equation}\label{#1}}
\newcommand{\ee}{\end{equation}}
\newcommand{\calF}{\mathcal{F}}
\newcommand{\calL}{\mathcal{L}}
\newcommand{\calN}{\mathcal{N}}
\newcommand{\calM}{\mathcal{M}}
\newcommand{\hatF}{\mathcal{F}_0}
\newcommand{\calG}{\mathcal{G}}
\newcommand{\calD}{\mathcal{D}}
\newcommand{\calE}{\mathcal{E}}
\newcommand{\hatpsi}[1]{\psi_{0,#1}}
\newcommand{\hatpsid}{\hatpsi{d}}
\newcommand{\Ta}{t}
\newcommand{\taua}{\tau_t}
\newcommand{\Tbeta}{\lambda}

\def\ip<#1,#2>{\left<#1,#2\right>}
\begin{document}
\newtheorem{theorem}{Theorem}[section]
\newtheorem{conjecture}[theorem]{Conjecture}
\newtheorem{problem}{Problem}
\newtheorem{definition}{Definition}
\newtheorem{proposition}[theorem]{Proposition}
\newtheorem{corollary}[theorem]{Corollary}
\newtheorem{lemma}[theorem]{Lemma}
\maketitle
\begin{abstract}
We investigate the possibility that the conformal and conformal thin sandwich (CTS)
methods can be used to parameterize the set of solutions of the vacuum Einstein constraint
equations.  To this end we develop a model problem obtained
by taking the quotient of certain symmetric  data on conformally flat tori.  
Specializing the model problem to a three-parameter family of conformal data
we observe a number of new phenomena for the conformal and CTS methods.  
Within this family, we obtain a general existence theorem so long as the mean
curvature does not change sign.  
When the mean curvature changes sign, we find that for certain data
solutions exist if and only if the transverse-traceless tensor is sufficiently small.
When such solutions exist, there are generically more than one.
Moreover, the theory for mean curvatures changing sign is shown to be extremely sensitive with respect 
to the value of a coupling constant in the Einstein constraint equations.
\end{abstract}
\def\Reals{\mathbb{R}}
\def\Nats{\mathbb{N}}
\section{Introduction}
Initial data for the Cauchy problem of general relativity consist of a 
Riemannian manifold and a second
fundamental form that satisfy a system of nonlinear PDEs known as the Einstein 
constraint equations.  A longstanding goal has been to find a constructive description
of the full set of solutions of these equations on a given manifold, and hence a method of producing all 
possible initial data.  Although this problem remains open in general, 
the conformal method of Lichnerowicz and Choquet-Bruhat and York provides an elegant and complete solution to the problem of constructing constant-mean curvature (CMC) solutions.
For example, on compact manifolds the solutions of the Einstein constraint equations
are effectively parameterized by selection 
of conformal data consisting of a conformal class for the metric, 
a so-called transverse-traceless tensor, and a (constant) mean curvature.  
The conformal method  
can also be used to construct non-CMC solutions of the constraint equations,
but much less is known in this case. Ideally one would like to show that selection of generic 
conformal data leads to a unique corresponding solution of the constraint equations.

Until recently, virtually all results for the conformal method only applied to near-CMC initial data.  The first construction using the conformal method of a family of initial data with 
arbitrarily specified mean curvature was given by Holst, Nagy, and Tsogtgrel in \cite{Holstetal08}.  Although this result represents a breakthrough for the conformal method,
it has a number of important limitations:
\begin{itemize}
	\item The near-CMC hypothesis is replaced by a smallness assumption on the transverse-traceless tensor (i.e. a small-TT hypothesis).
	\item It is not known if small-TT conformal data determine a unique solution.
	\item The construction only works on Yamabe-positive compact manifolds.
	\item The construction requires non-vanishing matter fields.
\end{itemize}
It was subsequently shown in \cite{Maxwell09} that the construction could be extended to vacuum initial data,
but the other restrictions remain.  These results are compatible with the possibility that
a large set of conformal data lead to no solutions or multiple solutions; from the point 
of view of parameterizing the full set of solutions one would like to show that this does 
not occur.

In this paper we investigate the conformal method and its variation,
the conformal thin sandwich (CTS) method, by studying a 
model problem obtained from a quotient of certain symmetric conformal data. 
Despite the  simplicity of the model problem, it captures the core 
issues of the conformal method, including the nonlinear coupling and difficulties regarding 
conformal Killing fields.  Moreover, the model problem is easily studied numerically, and thus
gives an important tool for suggesting theorems which might be proved in the future. 

We consider a three-parameter family of model conformal 
data that allow for simultaneous violations of both the near-CMC and small-TT 
conditions on a Yamabe-null manifold.  The mean curvatures in this family are written as the 
sum of an average mean curvature, $\Ta$, and a fixed zero-mean function 
describing departure from the mean.  If $\Ta$ is chosen so that the mean
curvature does not change
sign, we find that there exists a solution of the constraint equations so long
as the transverse-traceless tensor in the family is not identically zero.
When the mean curvature changes sign, the situation is more delicate.  We
observe in this regime non-existence for certain large transverse-traceless tensors,
non-uniqueness for certain small transverse traceless tensors, and a critical
value of $\Ta$ (depending on the choice of lapse function 
in the CTS method and the choice of conformal class representative in the standard 
conformal method) for which there is an infinite family of solutions when the 
transverse-traceless tensor vanishes identically.

Previous non-uniqueness results for the conformal method
have been obtained by adding separate, poorly behaved terms to the equations, either in the form of non-scaling matter  sources \cite{BOMP07,Walsh07} or from coupling with a separate 
PDE in the extended conformal thin sandwich method \cite{PY05,Walsh07}.  We prove here
the first nontrivial non-uniqueness result for the 
standard, vacuum conformal method.  It arises from the nonlinear coupling of the equations,
and indicates that the standard conformal and CTS methods already contain poorly behaved terms.

Intriguingly, we find that for mean curvatures in the three-parameter family with
changing sign, the existence theory depends sensitively on the 
values of the constants involved in the nonlinear coupling of the conformal method.  We show
that 
these constants are balanced in such a way that any arbitrarily small adjustment to their values 
lead to one of two different existence theories.  All previous results for the conformal method
depend only on the signs of the constants in these equations.  This sensitivity suggests
why it has been so difficult to obtain general large-data results for the conformal method.

The conformal data used in this study has one potential drawback: the mean curvature is not
continuous, but has jump discontinuities.  This level of regularity is lower than has
previously been considered for the fully coupled conformal method. We note, however,
that the CMC theory of the conformal method readily constructs solutions of the constraint
equations with certain kinds of discontinuous second fundamental forms (\cite{cb-low-reg,Maxwell05,Holstetal08}),
and we use the CMC results of \cite{cb-low-reg} to cope with the discontinuities in the mean curvature.
From this perspective the singularities in the mean curvature are comparatively mild. It seems likely, moreover, that the low
regularity techniques introduced in \cite{Maxwell05} and extended in \cite{Holstetal08} could be applied to the construction method of \cite{Maxwell09} to obtain results that
apply to non-CMC conformal data of the regularity we consider here. 
We will address this question in subsequent work.

\subsection{Conformal Parameterizations}
Let $(M^n,h)$ be a Riemannian manifold and let $K$ be a second fundamental form on $M^n$,
i.e. a symmetric $(0,2)$-tensor. The vacuum Einstein constraint equations for $(h,K)$
are
\begin{subequations}\label{eq:constraints}
\begin{alignat}{2}
R_{h}-\abs{K}_{h}^2+ \tr_{h} K ^2 &= 0 &\qquad&\text{[\small Hamiltonian constraint]}\label{eq:hamiltonian}\\
\div_{h} K -d\,\tr_{h} K &= 0 &&\text{[\small momentum constraint]}\label{eq:momentum}
\end{alignat}
\end{subequations}
where $R_{h}$ is the scalar curvature of $h$. 
For simplicity, we restrict our attention to compact manifolds.

\begin{problem}[Conformal Parameterization Problem]\label{prob:conf}
Let $(M^n,g)$ be a compact Riemannian manifold.  Find a constructive parameterization of the set of
solutions $(h,K)$ of equations \eqref{eq:constraints} such that $h$ belongs to the conformal class of $g$.
\end{problem}

If $(h,K)$ is a solution of equations \eqref{eq:constraints} with $h$ in the conformal class of $g$, we
may write $h=\phi^{q-2} g$ for some positive function $\phi$, where
\be
q=\frac{2n}{n-2}.
\ee
Without loss of generality we can write $K=\phi^{-2}\left( S + \frac{T}{n}g\right)$ where $S$ is a traceless
$(0,2)$-tensor and $T$ is a scalar field.  The constraint equations \eqref{eq:constraints}
for $(h,K)$ can then be written in terms of $(\phi,S,T)$ as
\begin{subequations}\label{conf-param-general}
\begin{alignat}{2}
-2\kappa q\Lap_g \phi + R_g \phi -\abs{S}_g^2 \phi^{-q-1} + \kappa T^2 \phi^{-q-1} &= 0 \label{eq:model-ham}\\
\div_g S - \kappa \phi^{q} d\; \left[\phi^{-q} T\right] &= 0 \label{eq:model-mom}
\end{alignat}
\end{subequations}
where
\be
\kappa = \frac{n-1}{n}.
\ee
The conformal parameterization problem then amounts to parameterizing the solutions $(\phi,S,T)$
of \eqref{conf-param-general}.

The conformal method \cite{cb-york-held} and its variation, the conformal thin sandwich (CTS) method \cite{YorkCTS99}, provide possible
approaches for solving Problem \ref{prob:conf}. An overview of these methods can be found
in \cite{BI04}. We summarize the techniques here to establish notation
and to state known results that impact our analysis of the model problem.

With the conformal method, one
specifies a mean curvature $\tau$ and a transverse-traceless tensor $\sigma$
(i.e. a symmetric, trace-free, divergence-free $(0,2)$-tensor).  We write
$T=\phi^{q}\tau$ and $S=\sigma + \ck W$ where $W$ is an unknown vector field and
$\ck$ is the conformal Killing operator defined by
\begin{equation}
(\ck V)_{ij} = \nabla_i V_j + \nabla_j V_i -\frac{2}{n} \nabla^k V_k g_{ij}.
\end{equation} 
Equations \eqref{conf-param-general} then become
\begin{subequations}\label{conf-param}
\begin{alignat}{2}
-2\kappa q\Lap_g \phi +R_g\phi -\abs{\sigma+\ck W}_g^2 \phi^{-q-1} + \kappa \tau^2 \phi^{q+1} &= 0 
&\quad&\text{[\small conformal Hamiltonian constraint]}\label{eq:cHamConst}\\
\div_g \ck W - \kappa \phi^{q} d\;\tau &= 0.
&\quad&\text{[\small conformal momentum constraint]}\label{eq:cMomConst}
\end{alignat}
\end{subequations}
These are coupled nonlinear elliptic equations to solve for
unknowns $(\phi,W)$.

For the CTS approach one specifies $\sigma$ and $\tau$ along with an additional
positive scalar function $N$ which represents a lapse.\footnote{Although
the CTS method is not usually presented as specifying $\sigma$ (compare \cite{YorkCTS99}) it is
straightforward to show that the presentation here is equivalent to the usual one.}
The CTS method is then obtained by replacing $\ck W$ with $\frac{1}{2N}\ck W$
wherever it appears in the discussion for the conformal method.
Although operationally similar to the conformal method, the CTS method has the advantage
of being conformally covariant. Specifically, if $\theta$ is a positive function,
then conformal data $(\theta^{q-2} g, \theta^{-2}\sigma, \theta^{q}N,\tau)$ yields the solution $(h,K)$
if and only if $(g,\sigma,N,\tau)$ does.  From the perspective of working with a fixed
background metric $g$, the standard conformal method simply corresponds to the CTS method with the
choice of $N=1/2$. We can think of the CTS approach as
providing many different parameterizations, one for each choice of $N$. It is not known if
certain choices of $N$ are superior for the purposes of finding a parameterization.
From the conformal covariance we observe that the choice of $N$ in the conformal-thin sandwich
method is equivalent to the choice of background metric for the conformal method: the solution
theory for the standard conformal method with the background metric $\hat g = \theta^{q-2} g$ is 
equivalent to the solution theory for the conformal thin sandwich method with lapse function
$N=\frac{1}{2}\theta^{-q}$. A conformal thin sandwich solution exists for $(g,\sigma,N,\tau)$ if and only if 
a standard conformal method solution exists for $(\hat g,\theta^{-2}\sigma, \tau)$, and 
the resulting solutions of the Einstein constraint equations are the same.

In the event that $\tau$ is constant, it is easy to see that the existence theory for
system \eqref{conf-param} reduces to the study of the Lichnerowicz equation
\begin{equation}\label{lich}
-2\kappa q\Lap_g \phi+R_g\phi -\abs{\sigma}_g^2 \phi^{-q-1} + \kappa \tau^2 \phi^{q+1} = 0.
\end{equation}
The obstruction to the existence of solutions of \eqref{lich} 
is stated in terms of the metric's Yamabe invariant
\be
\yamabe_g =  \inf_{f\in C^{\infty}(M)\atop f\not\equiv0} 
{\int_M 2\kappa q\abs{\nabla f}^2_g + R_g f^2\dV_g\over ||f||_{L^{q}}^2}
\ee
and we have the following theorem from \cite{const-compact}.
\begin{theorem}\label{thm:cmc}
Let $(M,g)$ be a smooth compact Riemannian manifold, let $\sigma$ be a transverse-traceless tensor,
and let $\tau$ be a constant.  Then there exists a positive solution of \eqref{lich} (and hence 
a solution of the conformally parameterized constraint equations \eqref{conf-param}) if and 
only if one of the following hold:
\begin{enumerate}
\item $\yamabe_g >0$, $\sigma\not\equiv 0$, 
\item $\yamabe_g =0$, $\tau\neq 0$, $\sigma\not\equiv 0$,
\item $\yamabe_g <0$, $\tau\neq 0$,
\item $\yamabe_g =0$, $\tau=0$, $\sigma\equiv 0$.
\end{enumerate}
When a solution exists it is unique, except in case 4) in which case any two solutions are related
by a positive scalar multiple.
\end{theorem} 
Hence the set of CMC solutions of \eqref{eq:constraints} having a metric conformally related to $g$ 
is essentially parameterized by choosing pairs $(\sigma,\tau)$.

The following non-CMC variation of Theorem \ref{thm:cmc} appeared in \cite{Maxwell09}.
\begin{theorem}\label{thm:non-cmc}
Let $(M,g)$ be a smooth compact Riemannian 3-manifold with no conformal Killing fields.
Suppose $\sigma$ and $\tau$ are a transverse-traceless tensor and a mean curvature such that 
one of the following hold:
\begin{enumerate}
	\item $\yamabe_g>0$, $\sigma\not\equiv 0$,
	\item $\yamabe_g=0$, $\sigma\not\equiv 0$, $\tau\not\equiv 0$
	\item $\yamabe_{g} < 0$ and there exists $\hat g$ in the conformal class of $g$ such that
  $R_{\hat g} = -\tau^2$.
\end{enumerate}
\textbf{If there exists a global upper barrier} for $(g,\sigma,\tau)$, then there exists at least
one solution of the conformally parameterized constraint equations \eqref{conf-param}.
\end{theorem}
The reader is referred to \cite{Maxwell09} for the definition of a global upper barrier 
(where it is called a global supersolution\footnote{The terminology global supersolution is perhaps misleading since it is not clear that all solutions of \eqref{conf-param} have associated global supersolutions.}); see also Appendix \ref{sec:even}.  
Cases 1-3 of Theorem \ref{thm:non-cmc}
reduce to those of Theorem \ref{thm:cmc} if $\tau$ is constant. 
Moreover, the condition on $\tau$ in Case 3 is necessary if $\yamabe_g<0$\cite{Maxwell05}.  
Until now, all results for the conformal method
are consistent with the possibility that (aside from the exceptional Case 4 of Theorem \ref{thm:cmc}), 
the conditions of Cases 1--3 of Theorem \ref{thm:non-cmc}
are necessary and sufficient for the unique solvability of equations \eqref{conf-param}.  
We show in this paper that this is not the case. In particular we find 
certain data satisfying the conditions of Case 2 for which there are 
nontrivially related multiple solutions.
We also find other symmetric data satisfying the conditions of Case 2 
for which there are no symmetric solutions
(and hence there are either no solutions or there are multiple solutions).

Global upper barriers can be found if the conformal data is CMC, satisfies a near-CMC condition 
such as
\be
{\max |\nabla \tau|\over \min |\tau|}\,\;\text{is sufficiently small},
\ee
or if $\yamabe_g>0$ and $\sigma$ is small-TT, i.e.
\be
\max |\sigma|_g\,\;\text{is sufficiently small, with smallness depending on $\tau$}.
\ee
This last upper barrier was first presented in \cite{Holstetal08} and led to the 
far-from CMC results of \cite{Holstetal08} and \cite{Maxwell09}.
Uniqueness theorems are available for a 
general class of near-CMC data under additional hypotheses on the size of 
$\abs{\nabla \tau}$ (\cite{IM96,ICA08}), but nothing is known concerning uniqueness in the small-TT case.

Results of O'Murchadha and Isenberg \cite{isenberg-omurchadha-noncmc} show that the condition $\sigma\not\equiv 0$
in hypotheses 1 and 2 of Theorem \ref{thm:near-cmc} 
is necessary for certain non-CMC data.  In particular, their
``no-go'' theorem proves that if $R_g\ge 0$ (or if $\yamabe_g \ge 0$ and we are using the
CTS method), then there does not exist a solution of \eqref{conf-param} if $\tau$ is near-CMC 
and $\sigma\equiv 0$.  
Rendall has also shown, as presented in \cite{isenberg-omurchadha-noncmc}, 
that there exists a class of Yamabe-positive
far-from CMC conformal data with $\sigma\equiv 0$ such that if a solution to
equations \eqref{conf-param} exists, it is not unique.  It is not known which of existence
or uniqueness fails for Rendall's data.

Symmetries pose a difficulty for the conformal method, and this
hampers the development of concrete examples.
Essentially all non-CMC existence results require that $(M^n,g)$ has no
conformal Killing fields.\footnote{\cite{CIM92} contains an exception, but it requires the conformal
data be constant along the integral curves of any conformal Killing fields. For the toroidal
initial data we consider in Section \ref{sec:roughfamily} this amounts to assuming
that $\tau$ is constant.}
Analytically this condition arises to guarantee that the operator $\div\ck$ is surjective,
but the need for this condition is more fundamental.
If $(M^n,g)$ admits a conformal Killing field $X$, then selection of a mean curvature
poses an a-priori restriction
on the solution $\phi$ of \eqref{conf-param} even before $\sigma$ is selected.  If $(h,K)$ is a 
solution of the constraint equations, then the mean curvature $\tau = \tr_h K$ must satisfy
$\int_M X(\tau)\; dV_h = 0$;
this identity is obtained by multiplying the momentum constraint \eqref{eq:momentum}
by $X$ and integrating by parts.
Writing this equation in terms of $g$ we find
\begin{equation}\label{CKFperp}
\int_M \phi^{q} X(\tau) \; dV_g =0,
\end{equation}
If $\tau$ is constant then equation \eqref{CKFperp} holds trivially.
If it is possible to find a solution $(\phi,W)$ for general data $(\sigma,\tau)$, then $W$
has to arise in such a way that $\phi$, which solves a Lichnerowicz equation depending on $W$,
also satisfies \eqref{CKFperp}.  The mechanism which might
cause this for arbitrary conformal data is not understood, and the issue is sidestepped
in the literature by assuming that there are no conformal Killing fields.

\section{Conformally flat $U^{n-1}$ symmetric data on $S^n$}
Let $S^1_r$ denote the circle of radius $r$ and let 
$M^n = S^1_{r_1}\times\cdots\times S^1_{r_n}$ with the product metric $g$. 
We can pick coordinates $x^k$ along each
factor such that $g_{ij} = \delta_{ij}$ and consider the following variation of Problem
\ref{prob:conf}.
\begin{problem} [Reduced Parameterization Problem]\label{prob:conf-model}
Find all solutions $(h,K)$ of the Einstein constraint equations on $M^n$ such that $h$ is conformally
related to $g$ and such that the Lie derivatives $\Lie_{\partial_k} g$ and $\Lie_{\partial_k} K$ vanish for
$1\le k \le n-1$.
\end{problem}
In practice we are seeking solutions such that $h$ and $K$ are periodic functions of $x^n$ alone;
by an obvious scaling argument we may reduce to the case $r_n=1$ and $x\equiv x^n\in [-\pi,\pi]$

The maximal globally hyperbolic spacetime obtained from such data will be a Gowdy
spacetime with a conformally flat Cauchy surface.  Our focus is not so much
to generate initial data for Gowdy spacetimes (the formulation of the constraint
equations found in \cite{PiotrGowdy} is more convenient for that purpose), but to use 
the conformally flat torus as a test case for
conformal parameterizations in general. 
We remark that the CMC version of Problem \ref{prob:conf-model} (including more general
toroidal background metrics) was effectively treated in \cite{Isenberg-diss}.

For the moment we work in three dimensions and use the variables
$(\phi,S,T)$ introduced in the previous section. In coordinates we can write
\be
S = \frac{1}{3}\left[\begin{matrix} -a-b & c & d \\
                           c  & -a +b & e\\
                           d  & e & 2a \end{matrix}\right].
\ee
Assuming that $S$ and $T$ are functions of $x=x^3$ alone we have
$\div S = \frac{1}{3}(d',e',2a')$ and hence the momentum constraint \eqref{eq:model-mom} reads
\be
\frac{1}{3}( d', e', 2a') = \frac{2}{3}\phi^6 (0,0, (\phi^{-6} T)').
\ee
Here primes to denote derivatives with respect to $x$.
Note that $S$ is transverse-traceless if and only if $a$, $d$, and $e$ are constant,
and that $(\phi,S,T)$ satisfies the momentum constraint if and only if $d$ and $e$
are constant and
\be
a' = \phi^6 (\phi^{-6} T)'.
\ee
Letting $\eta^2 = (b^2+c^2+d^2+e^2)/9$, and noting that $(M^n,g)$ is scalar flat,
the Hamiltonian constraint \eqref{eq:model-ham} reads
\begin{equation}
-8 \phi'' - 2\eta^2\phi^{-7} +\frac{2}{3}\left[ T^2-a^2\right]\phi^{-7} = 0.
\end{equation}

A similar derivation works in higher dimensions, and we obtain the reduced equations
\begin{equation}
\label{conf-reduced-general}
\begin{aligned}
-2\kappa q\;\phi'' - 2\eta^2\phi^{-q-1}+\kappa\left[ T^2-a^2\right]\phi^{-q-1} &= 0\\
a' - \phi^q (\phi^{-q} T)' &=0.
\end{aligned}\end{equation}
Solving Problem \ref{prob:conf-model} amounts to parameterizing the solutions $(\phi,\eta,a,T)$
of \eqref{conf-reduced-general}.

The conformal method can be described in this framework as follows. First we write
\be
T=\phi^{q}\tau
\ee
where $\tau$ is a prescribed mean curvature function and the conformal factor $\phi$ is unknown.  Additionally, we decompose
\be
a=\mu+w'
\ee
where $\mu$ is a prescribed constant and $w$ is an unknown function. The function $w$
is related to the vector field $W$ of the conformal method via $2W=w\partial_{n}$.
The  constant $\mu$ is part of the transverse-traceless tensor; to specify the remainder we 
select an arbitrary function $\eta$.
Equations \eqref{conf-reduced-general} become
\begin{equation}\label{conf-reduced}
\begin{aligned}
-2\kappa q\;\phi'' - 2\eta^2\phi^{-q-1}-\kappa(\mu+w')^2\phi^{-q-1}+\kappa \tau^2\phi^{q-1} &= 0\\
w'' - \phi^q \tau' &= 0.
\end{aligned}\end{equation}
For the CTS approach we additionally choose a positive function $N$ and write $a=\mu+1/(2N)w'$.
The CTS equations are then
\begin{equation}\label{CTS-reduced}
\begin{aligned}
-2\kappa q\;\phi'' - 2\eta^2\phi^{-q-1}-\kappa(\mu+(2N)^{-1}w')^2\phi^{-q-1}+\kappa \tau^2\phi^{q-1} &= 0\\
((2N)^{-1}w')' - \phi^q \tau' &= 0.
\end{aligned}\end{equation} 

Equations \eqref{CTS-reduced} provide a model for the full CTS equations on a Yamabe-null
manifold.  The nonlinear coupling for this system is the same as for the original equations.  Moreover,
the background metric on $S^1$ has a nontrivial conformal Killing field ($\partial_x$).  Hence
the central difficulties of the conformal method are present in the model.
Appendix \ref{sec:even} outlines how standard techniques for the conformal method can be adapted to equations \eqref{CTS-reduced} if the data satisfy an additional evenness hypothesis.  Our primary focus, however,
is on examining a family of conformal data for which we obtain stronger results than are possible with
the techniques of Appendix \ref{sec:even}.

\section{A family of low regularity conformal data}\label{sec:roughfamily}
The prescribed data for system \eqref{CTS-reduced} are a constant
$\mu$ and a function $\eta$ together with a mean curvature function $\tau$.
We will assume that $\eta$ is constant and work with a one-parameter
family of mean curvatures
\bel{tau-form}
\taua = \Ta + \Tbeta
\ee
where $\Ta$ is constant and 
\bel{beta-form}
\Tbeta(x) =\begin{cases}  -1 & -\pi<x<0 \\
                          1 & 0 < x <\pi.\\
\end{cases}
\ee

This three-parameter family is suitable for exploring
simultaneous violations of the near-CMC and small-TT hypotheses.
The parameters $\eta$ and $\mu$ control the size of the 
relevant pieces of the transverse-traceless tensor.  On the other hand, $\Ta$ controls the 
departure from CMC in the sense that for large values of $\Ta$ the mean curvature has small 
relative deviation from its mean, and is hence near-CMC (see also Proposition 
\ref{prop:doubling} and the subsequent discussion in  Appendix \ref{sec:even}).

Data of this kind fall outside the current theory of the conformal
method for two reasons. First, the manifold possesses a 
non-trivial conformal Killing field ($\partial_x$) and the non-CMC data is not constant along it.
Second, the discontinuities in $\taua$
make the data more singular than is treated in the current best low-regularity results 
of \cite{Holstetal08} for the full coupled system \eqref{conf-param}.  

We avoid both difficulties by showing that the reduced system \eqref{CTS-reduced} for this
data can be decoupled, and the analysis will reduce to the study of Lichnerowicz-type equations.
Just as for the CMC-conformal method, the decoupling removes potential obstructions posed 
by conformal Killing fields.  Moreover, the data we consider are only modestly irregular
for the Lichnerowicz equation alone.  In particular, the results of \cite{cb-low-reg} are applicable.

\subsection{Summary of results}\label{sec:summary}
We wish to solve
\bel{CTS-reduced-2}
\begin{aligned}
-2\kappa q\;\phi'' - 2\eta^2\phi^{-q-1}-\kappa(\mu+(2N)^{-1}w')^2\phi^{-q-1}+\kappa \taua^2\phi^{q-1} &= 0\\
((2N)^{-1}w')' - \phi^q \taua' &= 0
\end{aligned}
\ee
on $S^1$.  Here $N$ is a given smooth lapse function, $\eta$ and $\mu$ are constants, 
and $\taua$ is defined by \eqref{tau-form} and \eqref{beta-form}. We seek
solutions $(\phi,w)\in W_+^{2,p}(S^1)\times W^{1,p}(S^1)$ where $p>1$; the subscript 
$+$ denotes the subset of positive functions. An easy bootstrap argument shows that if
such a solution exists it belongs to $W^{2,\infty}_+(S^1)\times W^{1,\infty}(S^1)$.  If $(\phi,w)$
is a solution, so is $(\phi,w+c)$ for any constant $c$, and it determines the same 
solution of the constraint equations.  We will say that 
$(\phi,w)$ is the unique solution of \eqref{CTS-reduced-2} if any other solution
is of the form $(\phi,w+c)$.

The existence theory turns out to depend on the choice of lapse function $N$ in the conformal 
thin sandwich case (or equivalently, on the choice of conformal representative of the background metric 
in the standard conformal method). 
We define
\bel{gammaN}
\gamma_N=-\frac{\int_{S^1} \Tbeta N}{\int_{S^1} N}.
\ee
It is easy to see that $-1 < \gamma_N <1$ and that if $N$ is constant (as in the conformal
method with the flat background metric), then $\gamma_N=0$.
The near-CMC regime is expressed in terms of the distance between $\Ta$ and $\gamma_N$.
\begin{figure}
\begin{centering}
\includegraphics{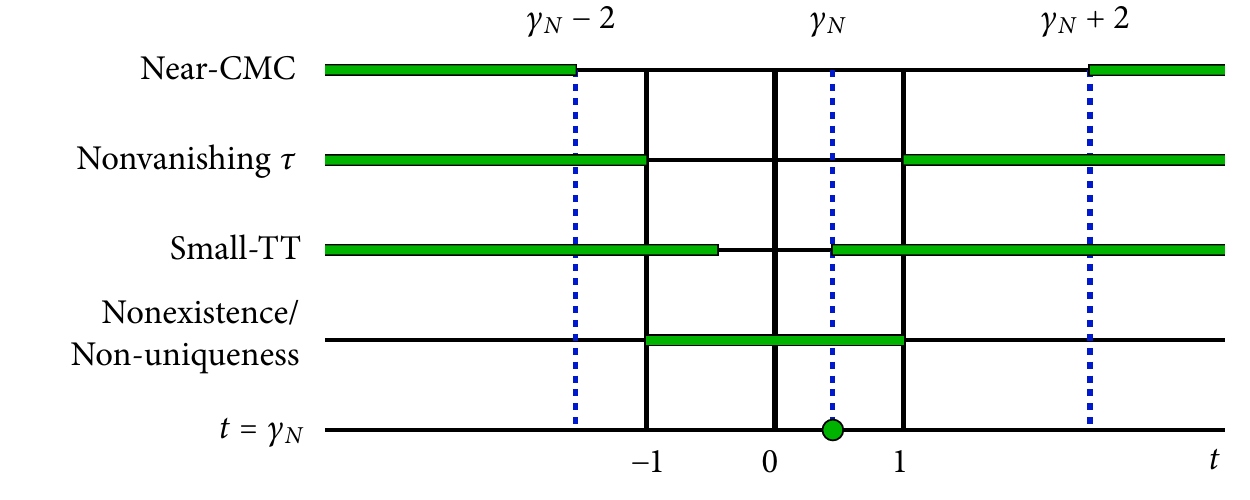}
\end{centering}
\caption{Ranges of $\Ta$ considered by the theorems of Section \ref{sec:summary}.}
\label{fig:Ta-range}
\end{figure}
\begin{theorem}[Near-CMC Results]\label{thm:near-cmc}
If $\abs{\Ta-\gamma_N}>2$ there exists a solution $(\phi,w)$ of \eqref{CTS-reduced-2} if and 
only if $\eta\neq 0$ or $\mu\neq 0$. Solutions are unique if $\mu=0$.
\end{theorem}
Note that the condition $\eta\neq 0$ or $\mu\neq 0$ is exactly the condition
that the transverse traceless tensor is not identically zero.  Hence Theorem
\ref{thm:near-cmc} extends the near-CMC existence/uniqueness theorem
of \cite{ICA08} and the ``no-go'' theorem of \cite{isenberg-omurchadha-noncmc} to this family of data.
We have not determined if uniqueness holds for $\mu\neq 0$.

The value $\Ta=\gamma_N$ is special, and we have the following result that is
a partial analogue of exceptional Case 4 of Theorem \ref{thm:cmc}.
\begin{theorem}[Exceptional Case: $\Ta=\gamma_N$]\label{thm:exceptional}
If $\Ta=\gamma_N$ and if $\mu=\eta=0$, then there exists a one-parameter
family of solutions of \eqref{CTS-reduced-2}.  If $\mu=0$ and $\eta\neq 0$, there does not exist
a solution.
\end{theorem}
It is not known if the non-existence result of Theorem \ref{thm:exceptional}
can be extended to include the case $\mu\neq 0$.

Given the non-existence result of Theorem \ref{thm:exceptional}, we can can only expect
a small-TT existence theorem if $\Ta\neq \gamma_N$. We have shown that if $\gamma_N=0$,
then this is essentially the only condition needed to obtain small-TT solutions, and 
have obtained a partial result for $\gamma_N\neq 0$.
\begin{theorem}[Small-TT Results]\label{thm:small-TT}
Suppose $\abs{\Ta}>\abs{\gamma_N}$ and $\abs{\Ta}\neq 1$.  If $\mu\neq 0$ or $\eta\neq 0$, and if $\mu$ and $\eta$
are sufficiently small, then there exists at least one solution of \eqref{CTS-reduced-2}.	
\end{theorem}
It is not known if existence holds if $\gamma_N\neq 0$ and either $\Ta=-\gamma_N$ or 
$\abs{\Ta}<\abs{\gamma_N}$.  The case $\abs{\Ta}=1$ remains open as well.

The mean curvature changes sign if and only if $\abs{\Ta}<1$.  We have the
following existence theorem that applies when $\abs{\Ta}>1$.  Note that
since $\abs{\gamma_N}<1$, the near-CMC condition $\abs{\Ta-\gamma_N}\ge 2$ is 
strictly stronger than the condition $\abs{\Ta}>1$.
\begin{theorem}[Non-vanishing Mean Curvature]\label{thm:a-gt-1}
Suppose $\abs{\Ta}>1$ and either $\mu\neq 0$ or $\eta\neq 0$.  Then there exists
at least one solution of \eqref{CTS-reduced-2}.
\end{theorem}
We have not determined if solutions are unique in this case, nor do we have 
an extension of the ``no-go'' theorem to this regime.  

The existence theory for $\abs{\Ta}<1$ is quite different than that for the near-CMC regime.
If $\mu=0$, we can show that when solutions exist, there are usually at least two, and
that if $\mu=0$ and $\eta$ is sufficiently large, then there are no solutions.  Hence a small-TT
hypothesis is necessary if $\mu=0$.
\begin{theorem}[Nonexistence/Non-uniqueness]\label{thm:critical-eta}
Suppose $\abs{\Ta}<1$ and $\mu=0$.  There exists a critical value $\eta_0 \ge 0$ such that
if $\abs{\eta}<\eta_0$ there exist at least two solutions of \eqref{CTS-reduced-2}, 
and if $\abs{\eta}\ge \eta_0$ there are no solutions.  If in addition
$\abs{\Ta}>\abs{\gamma_N}$, then $\eta_0>0$.
\end{theorem}

The preceding theorems omit the case $t=\pm 1$.  These values of $t$ are interesting
as they correspond to mean curvatures $\tau_t$ that are equal to zero on a large set.
The techniques for working with such mean curvatures are somewhat specialized, and
for simplicity we do not consider these values. We conjecture, however, that
Theorem \ref{thm:small-TT} can be extended to include $t=\pm1$.

The following theorem collects the results of Theorems \ref{thm:near-cmc} through \ref{thm:critical-eta}
specialized to the case $\mu=0$ and $\gamma_N=0$ where they are most complete.  
\begin{theorem}\label{thm:mu-zero} Suppose $\mu=0$ and $\gamma_N=0$.
\begin{enumerate}  
\item If $|\Ta|> 2$, there exists a solution a solution of \eqref{CTS-reduced-2}
if and only if $\eta\neq 0$. If a solution exists it is unique.
\item If $|\Ta|>1$ and $\eta\neq 0$, there exists at least one solution.
\item If $0<|\Ta|<1$, there is a critical value $\eta_0>0$.  If $0<\abs{\eta}<\eta_0$,
there are at least two solutions.  If $\abs{\eta}>\eta_0$ there are no solutions.
\item If $\Ta=0$ there exists a solution if and only if $\eta=0$, in which case
there is a one-parameter family of solutions.
\end{enumerate}
\end{theorem}
Figure \ref{fig:mu-zero} illustrates Theorem \ref{thm:mu-zero}.
We have a fairly complete picture of the existence/uniqueness theory when $\mu=0$;
we are missing a non-existence result for $0<\Ta<2$ if $\eta=\mu=0$, 
a uniqueness result for $1<|\Ta|<2$, and results for $\abs{\Ta}=1$.
\begin{figure}
\begin{center}
\includegraphics{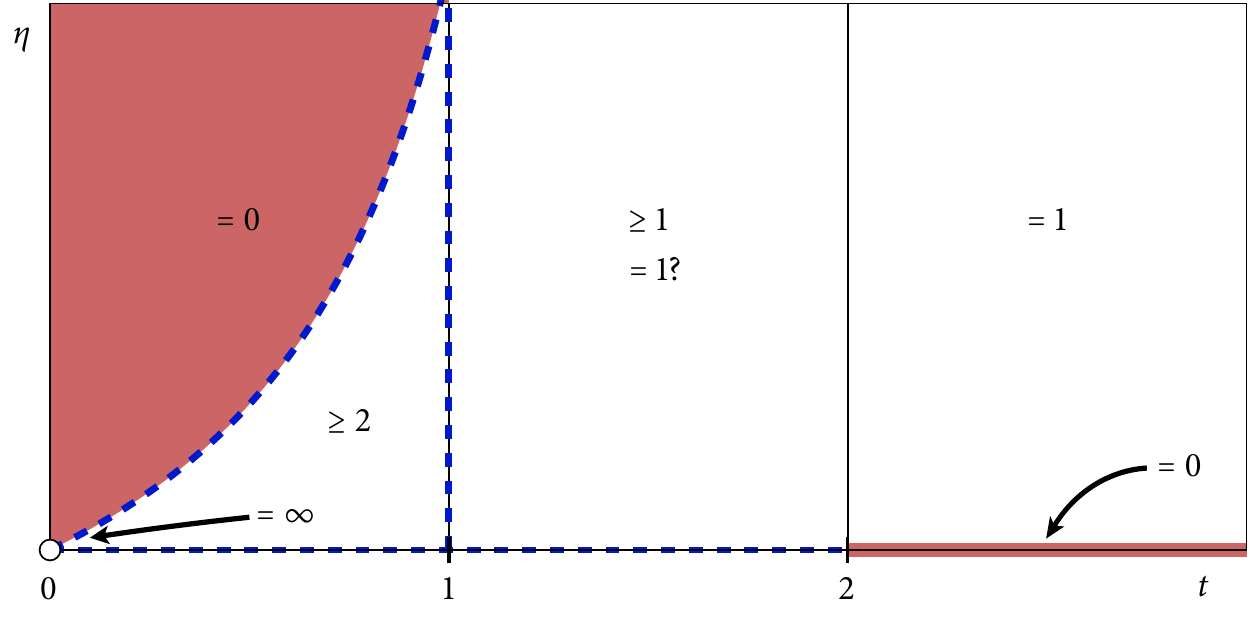}
\caption{Multiplicity of solutions for $\Ta\ge 0$ and $\eta\ge 0$ when $\mu=0$.  Dashed
lines correspond to curves where the multiplicity is unknown.  The shape of the curve separating the
existence and non-existence regions for $\Ta<1$ is conjectural.}
\label{fig:mu-zero}
\end{center}
\end{figure}

A little care is required in translating the results for the model problem to the full
conformal method. Because we are seeking solutions within a symmetry class,
the number of solutions we find is a lower bound for the total number of solutions.
Non-uniquness for the model problem implies non-uniquness for the full conformal method, but 
uniqueness only implies that there is a single solution with symmetry. Solutions without
symmetry (of which there must be more than one if there are any) may be present. Similarly,
non-existence for the model problem implies either non-existence or non-uniqueness for the 
full conformal method.

\subsection{Reduction to root finding}
In this section we show how for the specific choice of mean curvatures $\taua$
in equation \eqref{beta-form}, the existence theory of system \eqref{CTS-reduced-2}
can be reduced to the question of finding roots of a certain real valued function.

We first show that the solution of the momentum constraint can be determined exactly,
up to knowledge of the value of $\phi(0)$.
\begin{proposition}\label{prop:w-form}
Suppose $(\phi,w)\in W^{2,\infty}_+(S^1)\times W^{1,\infty}(S^1)$ is a solution 
of \eqref{CTS-reduced-2}.  Let
\be
\gamma_N=-\frac{\int_{S^1} \Tbeta N}{\int_{S^1} N}.
\ee
Then
\bel{w-solved}
\frac{1}{2N}w' = \phi(0)^q \left[\Tbeta+\gamma_N\right].
\ee
\end{proposition}
\begin{proof}
Notice that $\taua'=2\left[\delta_0-\delta_{\pi}\right]$ where $\delta_x$ denotes the Dirac delta distribution with singularity at $x$.  If $(\phi,w)$ is a solution of \eqref{CTS-reduced-2}, then
\bel{rough-momentum}
((2N)^{-1}w')' = 2\phi^q\,\left[\delta_0-\delta_{\pi}\right] = 2\phi(0)^q \delta_0-2\phi(\pi)^q\delta_\pi.
\ee
Since $\ip<((2N)^{-1}w')',1>=0$ (where $\ip<\cdot,\cdot>$ denotes the pairing of 
distributions on test functions) we have 
\be
0=\ip<\phi^q(\delta_0-\delta_\pi),1> = \phi(0)^q-\phi(\pi)^q.
\ee
Hence $\phi(0)=\phi(\pi)$. 

The momentum constraint then reads
\be
\left(\frac{1}{2N}w'\right)'= \phi(0)^q \Tbeta'.
\ee
Hence 
\bel{wisbeta-cts}
\frac{1}{2N}w' = \phi(0)^q \left[\Tbeta+C\right]
\ee
for some constant $C$.  Since $\int_{S^1} w'=0$ the value of $C$ is determined by
\be
\int_{S^1} 2N\,[\Tbeta+C] = 0.
\ee
This occurs precisely when $C=\gamma_N$.
\end{proof}

Substituting equation \eqref{w-solved} into the Hamiltonian constraint of system \eqref{CTS-reduced-2}
we obtain a nonlocal equation for $\phi$.
\begin{proposition}\label{prop:nonlocal}
Suppose $(\phi,w)\in W^{2,\infty}_+(S^1)\times W^{1,\infty}(S^1)$ solves \eqref{CTS-reduced-2}.  Then $\phi$ satisfies 
\begin{equation}\label{heaviside-const}
-2\kappa q\;\phi'' -2\eta^2\phi^{-q-1} -\kappa[\mu+\phi(0)^{q}(\gamma_N+\Tbeta)]^2\phi^{-q-1} +
\kappa(\Ta+\Tbeta)^2\phi^{q-1}=0.
\end{equation}

Conversely, suppose $\phi\in W^{2,\infty}_+(S^1)$ is a solution of \eqref{heaviside-const}.  
Then there exists a solution $w\in W^{1,\infty}(S^1)$ (uniquely determined up to a constant)
of \eqref{w-solved} and $(\phi,w)$ is a solution of \eqref{CTS-reduced-2}. 
\end{proposition}
\begin{proof}
If $(\phi,w)$ is a solution of \eqref{CTS-reduced-2} then 
Proposition \ref{prop:w-form} implies $w$ solves \eqref{w-solved}.  Substituting
this solution into the Lichnerowicz equation, we obtain equation \eqref{heaviside-const}.

Conversely, suppose $\phi$ solves \eqref{heaviside-const}.  
By the choice of $\gamma_N$, equation
\eqref{w-solved} is integrable and the solution $w\in W^{1,\infty}(S^1)$ is determined up to a constant.  
Let $w$ be such a solution.  By construction, $w$ solves the momentum constraint for $\phi$, and $\phi$
solves the Hamiltonian constraint for $w$.  That is, $(\phi,w)$ is a solution of \eqref{CTS-reduced-2}.
\end{proof}

To study the nonlocal equation \eqref{heaviside-const} we introduce a
family of Lichnerowicz equations depending on a positive parameter $d$:
\bel{lich-one-param}
-2\kappa q\;\phi_d'' - 2\eta^2\phi_d^{-q-1}- \kappa[\mu+ (\gamma_N+\Tbeta)d^q ]^2 \phi_d^{-q-1}
+\kappa(\Ta+\Tbeta)^2\phi_d^{q-1} = 0.
\ee
Clearly the solutions of \eqref{heaviside-const} are in one-to-one correspondence with
the solutions $\phi_d$ of \eqref{lich-one-param} satisfying $\phi_d(0)=d$.
The functions $\phi_d$ tend to grow as $d$ increases, and it will be more convenient to 
work with a rescaled function that is bounded as $d\ra\infty$.  The following
result follows easily from Proposition \ref{prop:nonlocal} after defining $\psi_d=d^{-1} \phi_d$.  We omit the proof.

\begin{proposition}\label{prop:psi_d}
The solutions of \eqref{CTS-reduced-2} are in one-to-one correspondence with the
functions $\psi_d\in W^{2,p}_+(S^1)$ satisfying
\bel{model-lich}
-2 \kappa q \dexp \psi_d'' - 2\eta^2 d^{-2q} \psi_d^{-q-1}- \kappa(\mu d^{-q}+\gamma_N+\Tbeta)^2 \psi_d^{-q-1}
+\kappa(\Ta+\Tbeta)^2\psi_d^{q-1} = 0
\ee
and
\be
\psi_d(0) = 1
\ee for some $d>0$.  Given a solution $\psi_d$ solving \eqref{model-lich}
and satisfying $\psi_d(0)=1$, the corresponding solution $\phi$ of \eqref{heaviside-const}
is $d\psi_d$.
\end{proposition}

Equation \eqref{model-lich} can be written as a Lichnerowicz equation of the form
\bel{lich-generic}
-u'' -\alpha^2 u^{-q-1} + \beta^2 u^{q-1} = 0
\ee
where $\alpha\not\equiv 0$ and $\beta\not\equiv 0$.  We have the following facts for this equation,
which are proved in Appendix \ref{sec:lich}.
\begin{proposition}\label{prop:lich-low-reg} 
Suppose $\alpha$ and $\beta$ in equation \eqref{lich-generic}
belong to $L^\infty(S^1)$ and that $\alpha\not\equiv 0$ and $\beta\not\equiv 0$. 
Let $p>1$.
\begin{enumerate}
\item There exists a unique solution $u \in W^{2,p}_+(S^1)$, and moreover $u\in W^{2,\infty}(S^1)$.
\item If $w\in W^{2,\infty}_+(S^1)$ is a subsolution of \eqref{lich-generic}, (i.e. $-w''-\alpha^2w^{-q-1}+\beta^2w^{q-1}\le 0$) then $w\le u$.
\item If $v\in W^{2,\infty}_+(S^1)$ is a supersolution of \eqref{lich-generic}, (i.e. $-v''-\alpha^2v^{-q-1}+\beta^2v^{q-1}\ge 0$) then $v\ge u$.
\item The solution $u\in W^{2,p}_+$ depends continuously on $(\alpha,\beta)\in L^\infty\times L^\infty$.
\end{enumerate}
\end{proposition}

We can now define the real valued function $\calF$ that will be the focus of our study.
\begin{definition}
Let $\Ta$ be a constant and let $\taua$ be defined by equations
\eqref{tau-form} and \eqref{beta-form}.  Let $N$ be a smooth lapse function
and let $\gamma_N$ be defined by equation \eqref{gammaN}.  Finally,
let $\eta$ and $\mu$ be constants.  For $d>0$, Proposition \ref{prop:lich-low-reg}
Part 1 implies that there exists a corresponding solution $\psi_d\in W^{2,\infty}_+(S^1)$ of
equation \eqref{model-lich}.  We define $\calF:\Reals_{>0}\ra \Reals_{>0}$ by
\bel{F-def}
\calF(d) = \psi_d(0).
\ee
We define $\hatF$ to be the analogous function corresponding to the same mean curvature but
vanishing transverse-traceless tensor (i.e. for $\mu=\eta=0$).
\end{definition}

From Proposition \ref{prop:psi_d} it is clear that the existence theory of the CTS method
for this family of data reduces to the study of the (algebraic) solutions of $\calF(d)=1$.
\begin{proposition}\label{prop:reduce-to-F}
The solutions $(\phi,w)\in W^{2,\infty}_+(S^1)\times W^{1,\infty}(S^1)$ of system \eqref{CTS-reduced-2} are in one to one correspondence with the positive solutions of $\calF(d)=1$.
\end{proposition}

\subsection{Solutions of $\calF(d)=1$}\label{sec:solutions}
\begin{figure}
\begin{center}
\includegraphics{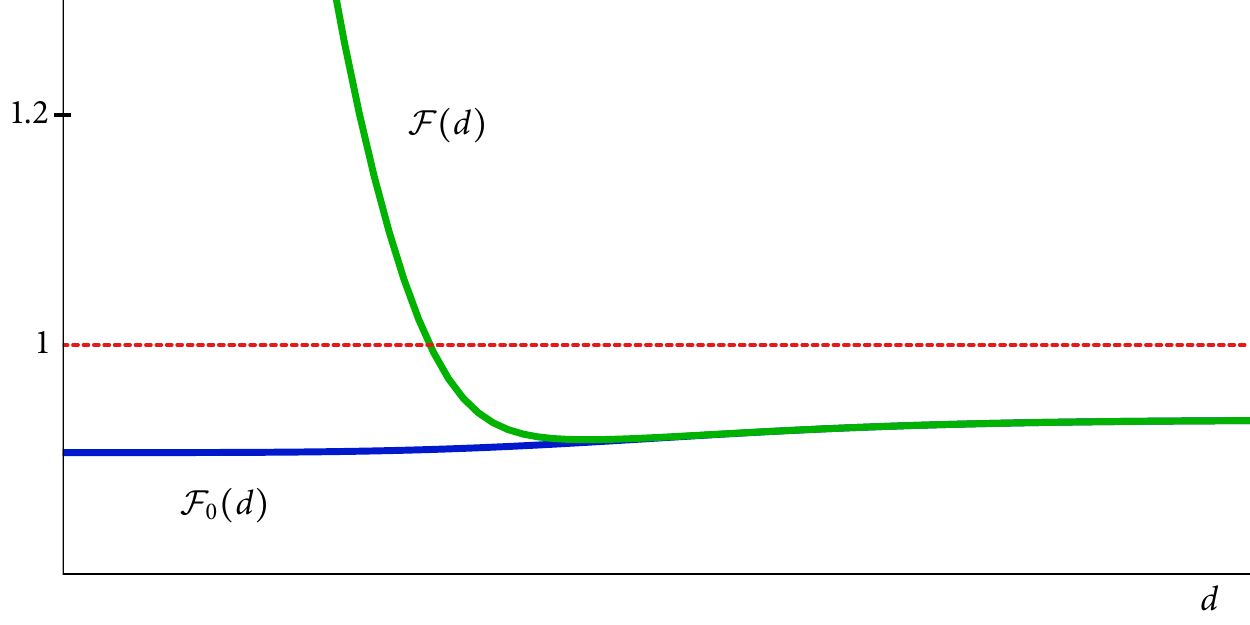}
\caption{Functions $\calF$ and $\calF_0$ for $\Ta=3/2$, $\mu=0$, and $\eta=3$.}
\label{fig:F}
\end{center}
\end{figure}
Theorems \ref{thm:near-cmc} through Theorem \ref{thm:critical-eta} follow from
Proposition \ref{prop:reduce-to-F} and  facts about $\calF$ and $\calF_0$ proved
in this section. 
Figure \ref{fig:F} shows representative 
graphs of $\calF$ and $\calF_0$ obtained by numerical 
computation for certain values of $\Ta$, $\eta$ and $\mu$.  Key features are the 
singular behaviour of $\calF$ at $d=0$, the limit of $\calF_0$
at $d=0$, and the common limits of $\calF$ and $\calF_0$ at $\infty$.  We note that
for the illustrated choice of $\Ta$, $\eta$ and $\mu$ it appears there is
exactly one solution of $\calF(d)=1$ and none of $\hatF(d)=1$.

\subsubsection{Elementary Estimates for $\calF$}
In this section we establish:
\begin{enumerate}
	\item If $\mu \neq 0$ or $\eta\neq 0$ (i.e. if the transverse-traceless tensor is not 
	identically zero) then
	$\calF(d)$ is $O(d^{-1})$ for $d$ sufficiently small.
	\item If $\mu=\eta=0$ then $\calF$ is uniformly bounded on $(0,\infty)$.
	\item For all values of $\mu$ and $\eta$, $\calF(d)$ is bounded above 
	      for values of $d$ sufficiently large.
	\item If $\mu\neq 0$ or $\eta\neq 0$, then a solution of $\calF(d)=1$ exists if and only if $\calF(d)\le1$ for some $d>0$.
\end{enumerate}

These facts are all demonstrated by examining constant sub- and supersolutions.
\begin{lemma}\label{lem:subsuper}
Suppose $\abs{\Ta}\neq 1$.  We define the constants
\be
\begin{aligned}
m_d &= \min(M_{d,+},M_{d,-})\\
M_d &= \max(M_{d,+},M_{d,-})
\end{aligned}
\ee
where 
\be
M_{d,\pm} = \left[ \frac{2\eta^2d^{-2q}+\kappa(\mu d^{-q}+\gamma_N\pm1)^2}{\kappa(\Ta\pm 1)^2}\right]^{\frac{1}{2q}}.
\ee
Then $m_d\le \psi_d\le M_d$
for all $d>0$ and in particular  
\be
m_d\le \calF(d)\le M_d.
\ee
\end{lemma}
\begin{proof}
A constant $M$ is a supersolution of \eqref{model-lich} so long as
\be
-2\eta^2 d^{-2q} M^{-q-1} - \kappa\left[\mu d^{-q}+\gamma_N+\Tbeta\right]^2 M^{-q-1} + 
\kappa (\Ta+\Tbeta)^2 M^{q-1} \ge 0.
\ee
Since $\Tbeta=\pm1$ on $S^1$, this is ensured if
\bel{sup-est}
\kappa(\Ta\pm1)^2 M^{2q} \ge 2\eta^2d^{-2q} +\kappa\left[\mu d^{-q}+\gamma_N\pm1\right]^2.
\ee
In particular, $M_d$ is a supersolution.  Proposition \ref{prop:lich-low-reg} Part 3 now implies
that $\psi_d\le M_d$ on $S^1$.  

A similar proof shows that $m_d$ is a subsolution if $\abs{\Ta}\neq 1$, and hence Proposition \ref{prop:lich-low-reg} Part 2 implies $\psi_d>m_d$ on $S^1$.
\end{proof}

From the limiting behaviour of $m_d$ and $M_d$ as $d\ra \infty$ we have estimates
for $\psi_d$ (and hence $\calF(d)$) for large values of $d$.
\begin{lemma}\label{lem:Minfty}
Suppose $\abs{\Ta}\neq 1$.  Let
\be
M_\infty = \max\left[ \left|\frac{1-\gamma_N}{1-\Ta}\right|^{\frac{1}{q}},
\left|\frac{1+\gamma_N}{1+\Ta}\right|^{\frac{1}{q}}\right]
\ee
and
\be
m_\infty = \min\left[ \left|\frac{1-\gamma_N}{1-\Ta}\right|^{\frac{1}{q}},
\left|\frac{1+\gamma_N}{1+\Ta}\right|^{\frac{1}{q}}\right].
\ee
Given $\epsilon>0$, 
\be
m_\infty-\epsilon \le \psi_d \le M_\infty+\epsilon
\ee
holds for $d$ sufficiently large. If $\mu=\eta=0$ then 
\be
m_\infty \le \psi_d \le M_\infty
\ee
for all $d>0$.
\end{lemma}
\begin{proof}
We note that
\be
\lim_{d\ra\infty} M_d = M_\infty
\ee
and
\be
\lim_{d\ra\infty} m_d = m_\infty.
\ee
Hence the bounds $m_\infty-\epsilon \le \psi_d \le M_\infty+\epsilon$ hold for $d$ 
sufficiently large.  

If $\mu=\eta=0$, then $m_d=m_\infty$ and $M_d=M_\infty$ for all $d>0$, so
$m_\infty\le \psi_d \le M_\infty$ for all $d>0$.
\end{proof}

The singular or bounded behaviour of $\calF$ near zero follows from
the analogous behaviour of the associated sub- and supersolutions.
\begin{lemma}\label{lem:F-bounds}
Suppose $\abs{\Ta}\neq 1$. If $\eta=\mu=0$ then
\bel{F-upper}
 \calF(d) \le M_\infty
\ee
for all $d>0$. Otherwise there is a positive constant $c$ such that
\bel{F-lower-sing}
\calF(d) \ge c d^{-1}
\ee
for $d$ sufficiently small.
\end{lemma}
\begin{proof}
We note that if $\eta\neq0$ or $\mu\neq 0$ then $M_{d,+}$ and $M_{d,-}$ are both
$O(d^{-1})$ at $d=0$ and hence so is $m_d$. 
The uniform upper bound \eqref{F-upper} when $\mu=\eta=0$ was proved in Lemma \ref{lem:Minfty}.
\end{proof}

The singularity of $\calF$ at $d=0$ gives a simple test for determining
if there is at least one solution of $\calF(d)=1$.
\begin{lemma}\label{lem:upper-is-enough} Suppose $\eta\neq 0$ or $\mu\neq 0$.  There
exists a solution of $\calF(d)=1$ if and only if for some $d>0$, $\calF(d)\le1$.
\end{lemma}
\begin{proof}
By Lemma \ref{lem:F-bounds}, $\calF(d)>1$ for $d$ sufficiently small.
Fixing $p>1$, from Proposition \ref{prop:lich-low-reg} Part 4 it follows that the map $d\mapsto\psi_d$
from $(0,\infty)$ to $W^{2,p}(S^1)$ is continuous.  From the continuous imbedding
of $W^{2,p}(S^1)\hookrightarrow C(S^1)$ it follows that $\calF$ is continuous and
the result now follows from the Intermediate Value Theorem.
\end{proof}

\subsubsection{Proof of Theorem \ref{thm:near-cmc} (Near-CMC Results)}
In this section we show that in the near-CMC regime ($\abs{\Ta-\gamma_N}>2$) 
the following hold:
\begin{enumerate}
	\item $\limsup_{d\ra\infty} \calF(d) < 1$.
	\item $\calF$ is differentiable and $F'(d)<0$ if $F(d)=1$ (and $\mu=0$).
	\item $\calF(d)<1$ for all $d$ if  $\mu=\eta=0$.
\end{enumerate}
The existence of a solution of $\calF(d)=1$ if $\mu\neq 0$ or $\eta\neq 0$ 
follows from Fact 1 and Lemma \ref{lem:upper-is-enough}.
The uniqueness of solutions of $\calF(d)=1$ if $\mu=0$ follows from Fact 2.  And the non-existence of
solutions of $\calF(d)=1$ if $\mu=\eta=0$ follows from Fact 3.

The upper bounds of Facts 1 and 3 follow from the constant supersolutions of Lemma \ref{lem:subsuper}. 
If effect, $\calF(d)<1$ because $\psi_d<1$ everywhere.  

\begin{lemma}\label{lem:near-cmc-control}
Suppose $\abs{\Ta-\gamma_N}>2$. Then 
\be
M_\infty < 1.
\ee
\end{lemma}
\begin{proof}
Note that since $\abs{\gamma_N}<1$, if $\abs{\Ta-\gamma_N}>2$, then
$\abs{\Ta}>1$ and in particular $\abs{Ta}\neq 1$.

Suppose first that $\Ta>1$.  Then
\be
M_\infty^q = 
\max\left( \abs{\frac{1-\gamma_N}{\Ta-1}},\abs{\frac{1+\gamma_N}{\Ta+1}}\right) =
\max\left( \frac{1-\gamma_N}{\Ta-1},\frac{1+\gamma_N}{\Ta+1}\right).
\ee
So $M_\infty<1$ if $1-\gamma_N<\Ta-1$ and $1+\gamma_N<\Ta+1$. The first equality holds since
$2 < \abs{\Ta-\gamma_N} = \Ta-\gamma_N$. The second holds since $\gamma_N<1<\Ta$.

The case where $\Ta<-1$ is proved similarly.	
\end{proof}

\begin{corollary}\label{cor:near-cmc-exist}
Suppose $\eta\neq 0$ or $\mu\neq 0$.  If
\be
\abs{\Ta-\gamma_N} > 2,
\ee
then there exists a solution of $\calF(d)=1$.
\end{corollary}
\begin{proof}
From Lemma \ref{lem:Minfty}, 
\be
\limsup_{d\ra \infty} F(d) \le M_\infty.
\ee
From Lemma \ref{lem:near-cmc-control} $M_\infty<1$.  Existence of a solution now follows from
Lemma \ref{lem:upper-is-enough}.
\end{proof}

If $\eta=0$ and $\mu=0$ (i.e. for vanishing transverse-traceless tensors) we have a
corresponding non-existence result which generalizes the ``no-go'' theorem of
\cite{isenberg-omurchadha-noncmc} to this family of data.  Recall that $\hatF$ corresponds to
$\calF$ with $\eta=\mu=0$.
\begin{corollary}\label{cor:no-go}
If 
\be
\abs{\Ta-\gamma_N} > 2
\ee
then $\hatF(d)<1$ for all $d>0$.  In particular, 
there are no solutions of $\hatF(d)=1$.
\end{corollary}
\begin{proof}
If $\eta=0$ and $\mu=0$ then $M_d=M_\infty$ for all $d>0$.  By
Lemma \ref{lem:near-cmc-control}, $M_\infty<1$.  Hence $\calF(d)<1$ for
all $d>0$.
\end{proof}

To show solutions of $\calF(d)=1$ are unique we show that $\calF$ is decreasing 
at any solution of $\calF(d)=1$.  We start by showing that $\calF$ is differentiable.
\begin{lemma}\label{lem:F-diff}
The function $\calF$ is differentiable. Moreover,
\be
\calF'(d) = h(0)
\ee
where $h\in W^{2,p}(S^1)$ solves
\be
-2\kappa q\; h'' +d^{q-2} V\, h = -R
\ee
and where
\bel{V-def}
V = 
	(q+1)\left[2\eta^2 d^{-2q}+ 
	\kappa(\mu d^{-q}+\gamma_N+\Tbeta)^2\right]\psi_d^{-q-2} + 
	(q-1)\kappa(\Ta+\Tbeta)^2\psi_d^{q-2}\\
\ee
and
\bel{R-def}
R =
(q+2)2\eta^2 d^{-q-3}\psi_d^{-q-1} 
+2q\kappa\mu d^{-3}(\mu d^{-q}+\gamma_N+\Tbeta)\psi_d^{-q-1} 
+(q-2)\left[ \kappa(\Ta+\Tbeta)^2\psi_d^{q-1} - \kappa(\mu d^{-q}+\gamma_N+\Tbeta)^2\psi_d^{-q-1}\right].
\ee
\end{lemma}
\begin{proof}
Consider the function $\calM:\Reals_{>0}\times W^{2,p}_+(S^1)\ra L^p(S^1)$ defined by
\be
\calM(d,\psi) = -2\kappa q\; \psi'' -2\eta^2d^{-q-2}\psi^{-q-1} 
-\kappa(\mu d^{-q}+\gamma_N+\Tbeta)^2d^{q-2}\psi^{-q-2}
s+\kappa(\Ta+\Tbeta)^2\psi^{q-1}.
\ee
Using the fact that $2q/n = q-2$ it follows that $\calM(d,\psi_d)=0$ for all $d>0$.

It is tedious but routine to show that $\calM$ is Fr\'echet differentiable
and
\be
\calM'[d,\psi](\delta,h)=  -2\kappa q\; h'' + V h + R\delta.
\ee
From the continuous embedding $W^{2,p}(S^1)\hookrightarrow C(S^1)$ it follows that the operators
$V$ and $R$ are continuous as functions of $\psi$ and $d$; see, for example, Lemma \ref{lem:fp}
below that can be used to show that they are locally Lipschitz.  
So the map $(d,\psi)\mapsto \calM'[d,\psi]$ is continuous.
The operator from $W^{2,p}(S^1)\ra L^p(S^1)$
\be
h\mapsto -2\kappa q\;h'' + V h 
\ee
has a continuous inverse as $V\in L^\infty \ge 0$ and $V\not\equiv 0$ 
(see, e.g. \cite{cb-low-reg} Theorem 7.7). The Implicit Function Theorem (\cite{NFA} Corollary 4.2)
then implies that given a solution of $\calM(d_0,\psi_0)=0$ there is a unique 
function $G$ defined near $d_0$ such that $\calM(d,G(d))=0$, and $G$ is continuously differentiable.  
But $\calM(d,\psi_d)=0$ for all $d$, so by the uniqueness of $G$ we have $G(d)=\psi_d$.
Let $h=G'(d)$.  Then by the chain rule
\bel{h-eq}
0=\frac{\partial}{\partial d} \calM(d,G(d)) = -2\kappa\; h'' + V h + R.
\ee

Now $\calF(d) = \psi_d(0)$.  Since the evaluation map $\psi\mapsto \psi(0)$
is linear and continuous on $W^{2,p}(S^1)$, it follows that $\calF$ is continuously differentiable
and $\calF'(d) = G'(d)(0)$.  That is, $\calF'(d) = h(0)$ where $h$ solves \eqref{h-eq}.
\end{proof}

\begin{proposition}\label{prop:near-cmc-unique}
Suppose $\abs{\Ta-\gamma_N}>2$. If $\mu=0$ there exists at most one solution of $\calF(d)=1$.
\end{proposition}
\begin{proof}
Suppose $\calF(d)=1$.  We will show that $\calF'(d)<0$, and hence there can be at most one solution.

Consider the functions of a real variable $z$
\be
g_\pm(z) = -(\gamma_N\pm 1)^2z^{-q-1}+(\Ta\pm 1)^2z^{q-1}
\ee
and
\be
f_\pm(z) = -2\eta^2d^{-2q}z^{-q-1}+g_\pm(z).
\ee
Note that $g_\pm$ and $f_\pm$ are increasing in $z$ for $z>0$ and and $f_+(M_+)=f_-(M_-)=0$
(where $M_\pm$ is defined in Lemma \ref{lem:subsuper}).

Let $I_-=(-\pi,0)$ and $I_+=(0,\pi)$.  Then 
\bel{eq:phidec}
-2\kappa q\dexp\psi_d''+f_{\pm}(\psi_d) = 0
\ee
on $I_\pm$.  Since the coefficients of the differential equation \eqref{eq:phidec}
are constant on $I_\pm$, the function $\psi_d$ is smooth on these intervals.

Suppose without loss of generality that $M_-\ge M_+$. By Lemma \ref{lem:subsuper}
$M_+\le \psi_d \le M_-$ on $S^1$.  Since $g_+(M_+) \ge f_+(M_+)=0$, we have $g_+(\psi_d)\ge 0$
on $I_+$.

To show that $g_-(\psi_d)\ge 0$ on $I_-$ we use the near-CMC assumption.
Since $\psi_d\le M_-$ and $f_-(M_-)=0$, it follows from equation \eqref{eq:phidec}
that $\psi_d''\le 0$ on $I_-$.  Since $\psi_d(-\pi)=\psi_d(0)=1$,
it follows that $\psi_d\ge 1$ on $I_-$.  Since $g_-$ is increasing, we conclude that
\be
g_-(\psi_d) \ge g_-(1) = -(\gamma_N-1)^2+(\Ta-1)^2 = (\Ta-1)^2\left[ 1-\frac{(\gamma_N-1)^2}{(\Ta-1)^2}\right].
\ee
Now
\be
\frac{(\gamma_N-1)^2}{(\Ta-1)^2} \le M_\infty^{2q} < 1
\ee
by the definition of $M_\infty$ and Lemma \ref{lem:near-cmc-control}.  Hence
$g_-(\psi_d) > 0$ on $I_-$.

By Lemma \ref{lem:F-diff}, $\calF'(d)=h(0)$ where 
\be
-2\kappa q\;h'' + V h = -R
\ee
and where $V$ and $R$ are defined in equations \eqref{V-def} and \eqref{R-def}.  Since $\mu=0$,
\be
R = (q+2)2\eta^2d^{-2q}d^{-q-3}\psi_d^{-q-1}+(q-2)d^{q-3}\kappa g_\pm(\psi_d)
\ee
on $I_\pm$.  Since $g_\pm(\psi_d)\ge 0$ and $g_-(\psi_d)>0$ we conclude that
\be
R\ge 0,\qquad R\not\equiv 0.
\ee
Since $V\ge 0$ and $V\not\equiv 0$, the strong maximum principle (\cite{epde} Theorem 9.6)
then implies that $h<0$ on $S^1$. In particular, $\calF'(d)= h(0) < 0$.
\end{proof}

Corollaries \ref{cor:near-cmc-exist} and \ref{cor:no-go}, together
with Propositions \ref{prop:near-cmc-unique} and \ref{prop:reduce-to-F} imply
Theorem \ref{thm:near-cmc} -- in the near-CMC regime $\abs{\Ta-\gamma_N}>2$ 
there exists a solution of \eqref{CTS-reduced-2}
if and only if the TT-tensor is not identically zero.  If $\mu=0$ the solution
is unique.  Although we have not determined uniqueness if $\mu\neq0$, we note
that Proposition \ref{prop:near-cmc-unique} is the first uniqueness result
for the conformal method that does not make use of a bound for $|\nabla\tau|$.

\subsubsection{Proof of Theorem \ref{thm:exceptional} (Exceptional Case: $\Ta=\gamma_N$)}
The value $\Ta=\gamma_N$ is special. We have a partial result
that is parallel to the exceptional Case 4 of Theorem \ref{thm:cmc}.
\begin{lemma}\label{lem:exceptional}
Suppose $\Ta=\gamma_N$. If $\mu=\eta=0$ then $\calF(d)=1$ for all $d>0$ and hence there is
a one-parameter family of solutions.  On the other hand, if $\mu=0$ but $\eta\neq 0$, then
there are no solutions.
\end{lemma}
\begin{proof}
If $\Ta=\gamma_N$ and $\mu=\eta=0$ then the unique solution of \eqref{model-lich} is clearly $\psi_d=1$
(for any $d$).  Hence $\phi$ solves \eqref{heaviside-const} if and only if $\phi$ is a positive constant.

On the other hand, suppose that $\mu=0$ and $\eta\neq 0$. Then the constant $1$ is evidently a 
subsolution of \eqref{model-lich}, as is $1+\epsilon$ for $\epsilon$ sufficiently small. 
Hence $\calF(d)>1$ for all $d>0$.
\end{proof}

Theorem \ref{thm:exceptional} follows from Lemma \ref{lem:exceptional} and 
Proposition \ref{prop:reduce-to-F}.

\subsubsection{Proof of Theorem \ref{thm:small-TT} (Small-TT Results)}
In this section we wish to show that solutions of $\calF(d)=1$ 
exist for small, nonzero, transverse-traceless tensors (i.e. if $\mu$ and $\eta$
are small but not both zero). From Lemma \ref{lem:exceptional} we know that if 
$\Ta=\gamma_N$, then
there are no solutions of $\calF(d)=1$ when $\mu=0$ and $\eta\neq0$.  So
we cannot expect to find small-TT solutions if $\Ta=\gamma_N$.
We show here that if $\gamma_N=0$, then this is the only obstacle.
We also obtain a partial result for $\gamma_N\neq 0$, showing small-TT solutions exist if $\abs{\Ta}>\abs{\gamma_N}$.

Recall that $\hatF(d)=\hatpsi{d}(0)$ where $\hatpsi{d}$ is defined analogously to $\psi_d$, but using $\mu=\eta=0$. We will establish the following facts:
\begin{enumerate}
\item If $\abs{\gamma_N}<\abs{\Ta}$, then $\lim_{d\ra 0^+} \hatF(d) < 1$.
\item For any fixed $d>0$, $\calF(d)$ approaches $\hatF(d)$ as $\mu$ and $\eta$
approach zero.
\end{enumerate}
So if $\abs{\gamma_N}<\abs{\Ta}$, and if $\mu$ and $\eta$ are sufficiently small, there is
a $d$ such that $\calF(d)<1$.  If in addition $\mu\neq0$ or $\eta\neq 0$, then Lemma \ref{lem:upper-is-enough} implies that there is at least one solution of $\calF(d)=1$.

If $\mu\neq 0$ or $\eta\neq 0$, Lemma \ref{lem:F-bounds} shows that $\calF(d)\ra\infty$ as $d\ra0$,
but that this singularity is not present if $\eta=\mu=0$.
In this case, the term $-\dexp\Lap \hatpsi{d}$ dominates equation \eqref{model-lich}
as $d\ra 0$ and we expect the solutions to be nearly constant.  The following lemma
computes the value of this constant, which is less than one if $\abs{\gamma_N}<\abs{\Ta}$.
\begin{lemma}\label{lem:psi-to-const}
Let
\be
\Psi_{N,\Ta} = \left[\frac{1+\gamma_N^2}{1+\Ta^2}\right]^{\frac{1}{2q}}.
\ee
If $\abs{\Ta}\neq 1$ then
\be
\hatpsi{d} \xrightarrow[d\ra 0]{} \Psi_{N,\Ta}
\ee
in $W^{2,p}(S^1)$ and hence uniformly on $S^1$.  In particular 
\be
\lim_{d\ra 0^+} \hatF(d) = \Psi_{N,\Ta}.
\ee
\end{lemma}
\begin{proof}
Recall that $\hatpsi{d}$ is the solution of
\bel{lich-hat-2}
-2d^{-\frac{nq}{2}} \kappa q\;  \hatpsi{d}'' - \kappa(\gamma_N+\Tbeta)^2 \hatpsi{d}^{-q-1} +
\kappa(\Ta+\Tbeta)^2\hatpsi{d}^{q-1} =0.
\ee
From Lemma \ref{lem:Minfty} (since $\mu=\eta=0$),
$0<m_\infty \le \hatpsid \le M_\infty$ 
for all $d$, and consequently there exists a positive constant 
$C$ such that
\be
\abs{- \frac{1}{2q}(\gamma_N+\Tbeta)^2 \hatpsid^{-q-1} + \frac{1}{2q}(\Ta+\Tbeta)^2\hatpsid^{q-1}} \le C
\ee
for all $d>0$.  
Since $\hatpsid$ satisfies \eqref{lich-hat-2} it follows that
\be
||\hatpsid''||_{L^p}^p \le 2\pi C d^{q-2}
\ee
for all $d>0$.  

The Poincar\'e inequality implies that there is a constant $c_p$ such that if $u\in W^{2,p}(S^1)$,
\be
||u||_{W^{2,p}(S^1)} \le c_p\left(||u''||_{L^p} + \left|\int_{S^1} u \right|\right).
\ee
Let $A_d=\frac{1}{2\pi}\int_{S^1} \hatpsi{d}$, and let $\epsilon_d = \hatpsi{d} - A_d$. 
Then
\be
||\epsilon_d||_{W^{2,p}(S^1)} \le c_p\left(||\epsilon_d''||_{L^p} + \left|\int_{S^1} \epsilon_d \right|\right)
= c_p||\hatpsid''||_{L^p(S^1)}  \le (2\pi c_p C\; d^{2q/n})^{1/p} \ra 0
\ee
as $d\ra 0$. 

We will now show that $A_d\ra\Psi_{N,\Ta}$ as $d\ra 0$.  Since $\hatpsid = A_d+\epsilon_d$,
it then follows that that $\hatpsid \ra \Psi_{N,\Ta}$ in $W^{2,p}(S^1)$ as $d\ra 0$.

Let $(d_k)$ be any positive sequence converging to zero.  Since 
\be
m\le A_{d_k} \le M,
\ee
some subsequence $\{A_{d_{k_l}}\}$ converges to a constant $A\in [m,M]$. Moreover,
$\hatpsi{d_{k_l}}\ra A$ uniformly.

Then $\hatpsi{d_{k_l}} = A_{d_{k_l}}+\epsilon_{d_k} \ra A$ in $W^{2,p}(S^1)$.  
For all $d$, $\int_{S^1} \hatpsi{d}'' = 0$ and hence
\be
\int_{S^1} \kappa(\gamma_N+\Tbeta)^2 \hatpsi{d}^{-q-1}-\kappa(\Ta+\Tbeta)^2\hatpsi{d}^{q-1}=0.
\ee
Using the uniform convergence of $\hatpsi{d_{k_l}}$ to $A$ and the fact that
$0<m\le \hatpsid\le M$ for all $d$  we conclude that
\be
A^{-q-1}\int_{S^1} \kappa(\gamma_N+\Tbeta)^2 - A^{q-1}\int_{S^1} \kappa(\Ta+\Tbeta)^2 =0.
\ee
Now 
\be
\int_{S^1} (\gamma_N+\Tbeta)^2 = \pi\left[ (\gamma_N+1)^2+(\gamma_N-1)^2\right] = 2\pi\left[ 1+\gamma_N^2\right].
\ee
Similarly,
\be
\int_{S^1} (\Ta+\Tbeta)^2 = 2\pi\left[ 1+\Ta^2\right].
\ee
Hence
\be
A = \left[\frac{1+\gamma_N^2}{1+\Ta^2}\right]^{\frac{1}{2q}} = \Psi_{N,\Ta}.
\ee
The uniqueness of the limit $A$ now implies that $A_d \ra \Psi_{N,\Ta}$ as $d\ra 0$.  
\end{proof}

\begin{proposition}\label{prop:small-TT}
Suppose $\abs{\Ta}>\abs{\gamma_N}$ and $\abs{\Ta}\neq 1$. Then there exists at least one solution of $\calF(d)=1$ if
\begin{enumerate}
\item $\eta\neq 0$ or $\mu\neq 0$ and
\item $\abs{\eta}$ and $\abs{\mu}$ are sufficiently small.
\end{enumerate}
\end{proposition}
\begin{proof}
Since $\abs{\Ta}>\abs{\gamma_N}$ it follows that the constant $\Psi_{N,\Ta}$ from Lemma
\ref{lem:psi-to-const} is less than 1.  In particular, $\hatF(d)<1$ for $d$ sufficiently small.
Fix a particular value of $d$ such that this holds.  By Proposition \ref{prop:lich-low-reg} Part
4 it follows that $\calF(d)\ra\hatF(d)$ as $(\eta,\mu)\ra(0,0)$. In particular, $\hatF(d)<1$ if $\mu$
and $\eta$ are sufficiently small.  The existence result now follows from 
Lemma \ref{lem:upper-is-enough}.
\end{proof}

\subsubsection{Proof of Theorem \ref{thm:a-gt-1} (Non-vanishing Mean Curvature)}
\label{sec:nonvanish}

From the definition of the mean curvatures $\tau_t$, we see that $\tau_t$ has
constant sign if $\abs{t}>1$, but changes sign if $\abs{t}\le 1$.  In this
section we wish to show that there are solutions of $\calF(d)=1$ so long as 
$\tau_t$ has constant sign.  

Recall that our near-CMC existence result Corollary \ref{cor:near-cmc-exist}
was obtained by showing that $\psi_d(x)<1$ for all
$x\in S^1$ if $d$ is  sufficiently large and $\abs{\Ta-\gamma_N}>2$.  We only need to show, however, 
that $\calF(d)=\psi_d(0)<1$ if $d$ is sufficiently large. Section \ref{sec:asympt} contains an asymptotic analysis
that allows us to compute the exact value of  $\lim_{d\ra\infty} \calF(d)$ (as well
as the speed of the convergence).
Assuming the results of Section \ref{sec:asympt} for now, we show in this section that:
\begin{enumerate}
\item $\lim_{d\ra\infty} \hatF(d) < 1$ if $\abs{t}>1$.
\item $\lim_{d\ra\infty} \hatF(d) = 1$ if $\abs{t}<1$.
\item $\lim_{d\ra\infty} \calF(d) = \lim_{d\ra\infty} \hatF(d)$ if $\abs{\Ta}\neq 1$.
\end{enumerate}
In particular, if $\abs{t}>1$, then $\calF(d)<1$ for some $d>0$.  If $\mu\neq 0$ or $\eta\neq0$,
then Lemma \ref{lem:upper-is-enough} then implies that there is a solution of $\calF(d)=1$.

\begin{definition}\label{def:rapid}
We say that $f(x)\ra L$ {\bf rapidly} at infinity if 
\be
\lim_{x\ra\infty} \abs{f(x)-L}x^n = 0
\ee
for all $n\in \Nats$.  

We say that $f(x)\ra L$ {\bf rapidly} at 0 if 
\be
\lim_{x\ra0} \abs{f(x)-L}x^{-n} = 0
\ee
for all $n\in \Nats$.  
\end{definition}

Recall that $\hatF(d)=\hatpsi{d}(0)$ where $\hatpsi{d}$ is defined analogously to $\psi_d$,
but with $\eta=\mu=0$.
\begin{proposition}\label{prop:lich-limit-hat}
Assume that $\abs{\Ta}\neq 1$.
Then 
\be
\lim_{d\ra\infty} \hatpsi{d}(0) = 
\begin{cases} 
1 & \abs{\Ta}<1 \\
\abs{\Ta}^{-\frac{1}{q}} & \abs{\Ta}> 1,
\end{cases}
\ee
and this convergence is rapid.
\end{proposition}
\begin{proof}
Assuming the results of Section \ref{sec:asympt}, it follows 
from Theorem \ref{thm:lich-limit}
applied to equation \eqref{lich-hat-2} (taking $\epsilon=\sqrt{2\kappa q} d^{-\frac{q}{n}}$) that
\be
\lim_{d\ra\infty} \psi_d(0) = 
\left[ 
\frac{ \kappa\abs{\gamma_N+1}+\kappa\abs{\gamma_N-1}} {\kappa\abs{\Ta+1}+\kappa\abs{\Ta-1}}
\right]^{\frac{1}{q}}
\ee
and this convergence is rapid. Note that since $\abs{\gamma_N} < 1$, 
$\abs{1+\gamma_N}+\abs{1-\gamma_N} = 2$.  
If $\abs{\Ta}<1$ then $\abs{1+\Ta}+\abs{1-\Ta} = 2$, otherwise
$\abs{1+\Ta}+\abs{1-\Ta} = 2\abs{\Ta}$.  The result now follows.
\end{proof}

We would like to establish a corresponding limit without the hypothesis $\eta=\mu=0$.
For large values of $d$ the contribution of the terms involving $\eta$ and $\mu$
in equation \eqref{model-lich} are small.
So we expect that $\hatpsi{d}$ should be a good approximation for $\psi_d$, and
we expect to obtain the same limit.
To make this idea precise, we will show that small perturbations of $\hatpsi{d}$ are sub-
and supersolutions of the equation for $\psi_d$. 

Recall from Lemma \ref{lem:Minfty} that $0<m_\infty\le\hatpsi{d}\le M_\infty$
for all $d>0$.  We define
$\calG_d:[-m_\infty/2,M_\infty] \ra L^\infty$ by 
\be
\calG_d(K) = \calN_d(\hatpsid+K).
\ee
where $\calN_d: W^{2,p}_+(S^1)\ra L^p(S^1)$ is the nonlinear Lichnerowicz operator 
\be
\calN_d(w) = 
-2\kappa q \dexp  w'' - 2\eta^2 d^{-2q} w^{-q-1}- \kappa[\mu d^{-q}+\Tbeta+\gamma_N  ]^2 w^{q-1}.
\ee
So $\hatpsid+K$ is a sub- or supersolution of \eqref{model-lich} if and only if $\calG_d(K) \le 0$
or $\ge 0$ almost everywhere.  

Using the fact that $\hatpsi{d}$ solves equation \eqref{lich-hat-2}
we can write
\be
\calG_d(K) = \calD(K)+\calE(K)
\ee
where
\bel{calDE}
\begin{aligned}
\calD(K) &= (\Ta+\Tbeta)^2\left[ (\hatpsid+K)^{q+1}-\hatpsid^{q+1}\right] \\
\calE(K) &= (\gamma_N+\Tbeta)^2\hatpsid^{-q-1}-\left[2\eta^2d^{-2q}+(\mu d^{-q} + \gamma_N+\Tbeta)^2)\right]^2(\hatpsid+K)^{-q-1}.
\end{aligned}
\ee

\begin{lemma}\label{lem:FTC}
There exist positive constants $D_-$, $D_+$, $E_-$ and $E_+$ such that
\bel{ineq:E}
\begin{alignedat}{2}
E_- K & \;\le\; (\gamma_N+\Tbeta)^2 \left[ \hatpsid^{-q-1}-(\hatpsid+K)^{-q-1}\right] &\;\le\; E_+ K &\quad\;\; K\ge 0\\
E_+ K & \;\le\; (\gamma_N+\Tbeta)^2 \left[ \hatpsid^{-q-1}-(\hatpsid+K)^{-q-1}\right] &\;\le\; E_- K &\quad\;\; K\le 0.
\end{alignedat}
\ee
and
\bel{ineq:D}
\begin{alignedat}{2}
D_- K &\;\le\; \calD(K)&\;\le\; D_+ K &\quad\;\; K \ge 0\\
D_+ K &\;\le\;  \calD(K) &\;\le\; D_- K &\quad\;\; K \le 0\\
\end{alignedat}
\ee
for all $d>1$ and all $K\in [-m_\infty/2,M_\infty]$.
\end{lemma}
\begin{proof}
First consider the expression $f_A(h)=A^{-q-1}-(A+h)^{-q-1}$ for $A\in [m_\infty,M_\infty]$ and $h\in[-m_\infty/2,M_\infty]$.
Then
\be
f_A(h) = \int_0^1 (q+1) (A+th)^{-q-2}\;dt\; h.
\ee
If $h\ge 0$ then 
\be
(q+1)(2M_\infty)^{-q-2} h \le f_A(h) \le (q+1)(m_\infty/2)^{-q-2} h.
\ee
If $h\le 0$ then 
\be
(q+1)(m_\infty/2)^{-q-2} h \le f_A(h) \le (q+1)(2M_\infty)^{-q-2} h.
\ee
Inequalities \eqref{ineq:E} now follow letting $E_+ =\max[(\gamma_N-1)^2,(\gamma_N+1)^2] (q+1) (m_\infty/2)^{-q-2}$
and $E_-=\min[(\gamma_N-1)^2,(\gamma_N+1)^2] (q+1) (2M_\infty)^{-q-2}$.

The argument for inequality \eqref{ineq:D} is similar.
\end{proof}

\begin{proposition}\label{prop:lich-limit-same}
There exists a constant $c>0$ such that
\bel{limit-order}
|| \hatpsid-\psi_d||_\infty < c d^{-q}
\ee
for all $d$ sufficiently large.  In particular, 
\be
\lim_{d\ra\infty} \calF(0) = \lim_{d\ra\infty} \psi_d(0) = \lim_{d\ra\infty} \hatpsid(0).
\ee
\end{proposition}
\begin{proof}
For each $d$ sufficiently large, we will find constants $K_-(d)$ and $K_+(d)$ 
that are $O(d^{-q})$ and that satisfy $\calG_d(K_-(d))<0$ and $\calG_d(K_+(d))>0$.
Assuming this for the moment, we see that 
$\hatpsid + K_-(d)$ and $\hatpsid+K_+(d)$ are sub- and supersolutions of \eqref{model-lich}
and hence $\hatpsid+K_-(d)\le \psi_d \le \hatpsid+K_+(d)$ for $d$ sufficiently large.
The asymptotics of $K_\pm(d)$ then imply inequality \eqref{limit-order}.
		
Notice that $\calD(K)$ has the same sign as $K$. So $\calG_d(K)>0$ if
$K>0$ and $\calE(K)>0$.  Now if $0<K\le M_\infty$ then Lemma \ref{lem:FTC} implies
\be
\begin{aligned}
\calE(K) &= (\gamma_N+\Tbeta)^2[\hatpsid^{-q-1}-(\hatpsid+K)^{-q-1}] -
\left[2\eta^2 d^{-2q}+(\mu d^{-q}+\Tbeta+\gamma_N)^2)\right]^2(\hatpsid+K)^{-q-1} \\
& \ge {E_-} K - \left[(2\eta^2 +\mu^2)d^{-2q} + 4|\mu|d^{-q}\right] (m_\infty/2)^{-q-1}.
\end{aligned}
\ee
Let 
\be
K_+(d) = \frac{\left[(2\eta^2 +\mu^2)d^{-q} + 4|\mu|\right] (m_\infty/2)^{-q-1} }{E_-} d^{-q}.
\ee
Then $0<K_+(d)\le M_\infty$ if $d$ is sufficiently large, and we have $\calE(K_+(d))\ge 0$
and $\calG_d(K_+(d))\ge 0$ also.

On the other hand, if $-m_\infty/2\le K<0$ then Lemma \ref{lem:FTC} implies
\be
\begin{aligned}
\calE(K) &\le E_- K -(2\eta^2 +\mu^2)d^{-2q}(2M_\infty)^{-q-1} + 4|\mu|d^{-q} (m_\infty/2)^{-q-1}.
\end{aligned}
\ee
Let
\be
K_-(d) = \frac{4|\mu|(m_\infty/2)^{-q-1} }{E_-} d^{-q}.
\ee
Then $-m_\infty\le K_-(d)<0$ if $d$ is sufficiently large. We then have $\calE(K_+(d))\le 0$
and $\calG_d(K_-(d))\le 0$ also.

Since $K_-(d)$ and $K_+(d)$ are both $O(d^{-q})$. we have proved the desired result.
\end{proof}

We now summarize the argument that, along with Proposition \ref{prop:reduce-to-F}, 
proves Theorem \ref{thm:a-gt-1}.
\begin{proposition}\label{prop:non-vanishing-tau}
Suppose $\abs{\Ta}>1$.  If $\eta\neq 0$ or $\mu\neq 0$, there exists at least one solution
of $\calF(d)=1$.
\end{proposition}
\begin{proof}
By Propositions \ref{prop:lich-limit-hat} and \ref{prop:lich-limit-same}, if $\abs{\Ta}>1$ then
\be
\lim_{d\ra\infty} \calF(d) =\left[ \frac{1}{|\Ta|}\right]^{\frac1 q} < 1.
\ee
So $\calF(d)<1 $ for $d$ sufficiently large.  Since $\eta\neq 0$ or $\mu\neq 0$, Lemma
\ref{lem:upper-is-enough} now implies there exists a solution of $\calF(d)=1$.
\end{proof}

\subsubsection{Proof of Theorem \ref{thm:critical-eta} (Nonexistence/non-uniqneness)}
\label{sec:critical}
In this section we restrict our attention to the case $\mu=0$, so that $\eta$ alone controls
the size of the transverse-traceless tensor. We show that if $\abs{\Ta}<1$, (i.e. when $\tau_t$ changes sign),
then there is a critical threshold $\eta_0\ge 0$ for the size of $\eta$.  If $\eta>\eta_0$,
then there are no solutions of $\calF(d)=1$, whereas if $\eta<\eta_0$ there are at least two.
In some cases we can show that $\eta_0>0$ and hence there are multiple 
solutions for small values of $\eta$.

The choice of $\eta$ plays a critical role in this section, so we use the notation 
$\calF_{[\eta]}$ to distinguish different functions $\calF$ corresponding
to different values of $\eta$. Since  $\calF_{[\eta]}(d)$ only depends on $\eta^2$,
we can assume that $\eta\ge0$. We will show the following facts (assuming $\mu=0$ and 
$\abs{t}<1$):
\begin{enumerate}
	\item $\lim_{d\ra\infty} \calF_{[\eta]}(d) = 1$, and this limit is approached from above.
	\item For any fixed $d>0$, $\calF_{[\eta]}(d)$ is strictly increasing in $\eta$.
	\item On any finite interval $(0,d_0]$ we can find $\eta$ sufficiently large so that $\calF_[\eta](d)>1$ on $(0,d_0]$.
\end{enumerate}
The idea of the proof proceeds as follows.
Picking an arbitrary $\eta>0$, Fact 1 implies
$\calF_{[\eta]}(d)>1$ for $d$ larger than some $d_0$.  Using Fact 3 we 
then increase $\eta$ to ensure that $\calF_{[\eta]}(d)>1$ on $(0,d_0]$.  Fact 2
ensures that after having increased $\eta$, we still have the condition $\calF_{[\eta]}(d)>1$
for $d>d_0$. So $\calF_{[\eta]}>1$ for all $d>0$ and there are no solutions of $\calF_{[\eta]}(d)=1$.
The existence of a critical value of $\eta$ follows from Fact 2:
if no solutions exist for some $\eta$, then $\calF_{[\eta]}(d)>1$ for all
$d$ and raising the value of $\eta$ maintains this inequality.  On the other hand, 
since $\calF_{[\eta]}(d)>1$ for $d$ large (by Fact 1) and for $d$ near zero (since $\calF_{[\eta]}$
is singular there), if $\calF_{[\eta]}(d)<1$ for some $d$, then there will be at least two solutions.

\begin{proposition}\label{prop:raising}
For fixed $d$, the value of $\calF_{[\eta]}(d)$ is strictly increasing in $\eta$. Moreover,
\bel{eq:eta-subsol}
\calF_{[\eta]}(d) \ge \left[\frac{2\eta^2}{\kappa(1+|\Ta|)^2}\right]^{\frac{1}{2q}}d^{-1}.
\ee
\end{proposition}
\begin{proof}
Fix $d>0$ and suppose $0\le \eta_1 \le \eta_2$.  Let $\psi_{d,1}$ and $\psi_{d,2}$ be the corresponding
solutions of \eqref{model-lich}.  Then substituting $\psi_1$ into the equation for $\psi_2$ we
have
\be
-2\kappa q \dexp \psi_{d,1}'' - 2\eta_2^2d^{-2q}\psi_{d,1} -\kappa[\mu d^{-q}+\Tbeta+\gamma_N]^2\psi_d^{-q-1} +\kappa(\Ta+\Tbeta)^2\psi_d^{q-1} =
2(\eta_1^2-\eta_2^2)d^{-2q}\psi_{d,1} < 0.
\ee
So $\psi_{d,1}$ is a subsolution of the equation for $\psi_{d,2}$ and $\psi_{d,1}\le \psi_{d,2}$.
A similar computation shows that 
$\psi_{d,1}+\epsilon$ is also a subsolution for $\epsilon>0$ sufficiently small an hence
$\psi_{d,1}<\psi_{d,2}$ everywhere.  In particular, $\calF_{[\eta_1]}(d) < \calF_{[\eta_2]}(d)$.

To obtain the estimate \eqref{eq:eta-subsol} we note that a constant $k$ is a subsolution of \eqref{model-lich} if
\be
-2\eta^2d^{-2q}k^{-q-1} + \kappa(\Ta+\Tbeta)^2k^{q-1} \le 0.
\ee
This holds in particular if 
\be
-2\eta^2d^{-2q}k^{-q-1} + \kappa(1+|\Ta|)^2k^{q-1} \le 0
\ee
and therefore if 
\be
k^{2q} = \frac{2\eta^2}{\kappa (1+|\Ta|)^2}d^{-2q}.
\ee
Since $\calF_{[\eta]}(d)\ge k$ if $k$ is a subsolution of \eqref{model-lich}, we have established
inequality \eqref{eq:eta-subsol}.
\end{proof}

\begin{proposition}\label{prop:force-from-below}
Suppose $\mu=0$ and $\eta\neq 0$.  
Then there exists a constant $c>0$ such that
\bel{lim-below}
\psi_d \ge \hatpsid+ c d^{-2q}
\ee
for all $d$ sufficiently large.
\end{proposition}
\begin{proof}
We use the function $\calG_d:[-m_\infty/2,M_\infty]\ra L^\infty$ defined in Section \ref{sec:nonvanish}.
Recall that $\hatpsi{d}+K$ is a subsolution of the equation for $\psi_d$ if $\calG_d(K)\le 0$
almost everywhere.  Recall also that $\calG_d$ can be written
\be
\calG_d(K) = \calD(K)+\calE(K)
\ee
where $\calD$ and $\calE$ are defined in equations \eqref{calDE}.  

If $0<K\le M_\infty$, then by Lemma \ref{lem:FTC}
\be
\calD(K) \le   D_+ K.
\ee
for a certain constant $D_+>0$.
Also,
\be
\begin{aligned}
\calE(K) &= (\gamma_N+\Tbeta)^2[\hatpsid^{-q-1}-(\hatpsid+K)^{-q-1}] -
2\eta^2 d^{-2q}(\hatpsid+K)^{-q-1} \\
&\le E_+ K - 2\eta^2 (2M_\infty)^{-q-1} d^{-2q} 
\end{aligned}
\ee
for a certain constant $E_+>0$.
Let 
\be
K_- = \frac{2\eta^2(2M_\infty)^{-q-1}}{D_+ + E_+} d^{-2q}.
\ee
If $d$ is sufficiently large, then $0<K_-\le M_\infty$ and we then have
\be
\begin{aligned}
	\calG_d(K_-) &= \calD(K_-)+\calE(K_-) \\
	&\le D_+ K_- + E_+ K_- - 2\eta^2 (2M_\infty)^{-q-1} d^{-2q} \\
	& = 0.
\end{aligned}
\ee
So $\hatpsi{d}+K_-$ is a subsolution, and we have obtained inequality \eqref{lim-below} with $c=2\eta^2(2M_\infty)^{-q-1}/(D_+ + E_+)$.
\end{proof}

The following Proposition formalizes the arguments made at the start of this section
and, along with Proposition \ref{prop:reduce-to-F}, completes the proof of Theorem
\ref{thm:critical-eta}.
\begin{proposition}\label{prop:critical-eta}
Suppose $\abs{\Ta}<1$ and $\mu=0$.  There exists $\eta_0\ge0$ such that if $0<\abs{\eta}<\eta_0$
there exists at least two solutions of $\calF(d)=1$, while if $\abs\eta>\eta_0$ there are no solutions.
If $\abs{\Ta}>\gamma_N$ then $\eta_0>0$.
\end{proposition}
\begin{proof}
We first show that $\psi_d(0) > 1$ for $d$ sufficiently large.  
From Proposition \ref{prop:lich-limit-same} we know that $\lim_{d\ra\infty}\hatpsid(0) = 1$ and that this convergence
is rapid.  On the other hand, from Proposition \ref{prop:force-from-below} there is
a positive constant $c$ such
that $\psi_d(0) > \hatpsid(0) + c d^{-2q}$.  Hence 
\be
\psi_d(0)-1 \ge (\hatpsid(0)-1) +c d^{-2q} = \left[(\hatpsid(0)-1)d^{2q}+c\right] d^{-2q}.
\ee
From the rapid convergence we have $(\hatpsid(0)-1)d^{2q}\ra 0$ as $d\ra \infty$ and
hence $\psi_d(0) > 1 $ for $d$ large enough.

To show that there are no solutions for $\eta$ sufficiently large, fix a given $\eta_1$ and
pick $d_0$ so that if $d>d_0$ then $\calF_{[\eta_1]}(d)>1$. From inequality
\eqref{eq:eta-subsol} we can find $\eta_2$ so that $\calF_{[\eta_2]}(d)>1$
for all $d\in (0,d_0]$.  Letting $\eta=\max(\eta_1,\eta_2)$, it follows from 
Proposition \ref{prop:raising} that $\calF_{[\eta]}(d)>1$ for all $d>1$.

Let $A=\inf\{\eta\ge 0: \calF_{[\eta]}(d)>1\quad\text{for all $d>0$}\}$; we have just
shown that $A$ is nonempty.
Suppose $\eta\in A$ and $\eta'\ge\eta$.  Proposition \ref{prop:raising} implies that
for any $d>0$, $\calF_{[\eta']}(d)\ge \calF_{[\eta]}(d)  >1$ and hence $\eta'\in A$.
Let $\eta_0=\inf A$.  If $\eta>\eta_0$ then $\eta\in A$ and there are no solutions
of $\calF_{[\eta]}(d)=1$.

Suppose $0<\eta<\eta_0$, and pick $\eta'$ so $\eta<\eta'<\eta_0$.  Then $\eta'\not\in A$
and for some $d_0$, $\calF_{[\eta']}(d_0)\le 1$.   
By Proposition \ref{prop:raising}, $\calF_{[\eta]}(d_0) < \calF_{[\eta']}(d_0) \le 1$.  From
Lemma  \ref{lem:F-bounds} we know that $\calF_{[\eta]}(d)>1$ for $d$ sufficiently small, and we 
have already shown that
$\calF_{[\eta]}(d)>1$ for $d$ sufficiently large.  From the continuity of $F$ it follows that
there are at least two solutions of $\calF_{[\eta]}(d)=1$, one for $d<d_0$ and one for $d>d_0$.

Proposition \ref{prop:small-TT} implies that $\eta_0>0$ if $\abs{\Ta}>\abs{\gamma_N}$;
if $\eta_0=0$ then there can only be solutions of \eqref{model-lich} if $\eta=0$.
\end{proof}

We have now proved all the results of Section \ref{sec:summary}, up to the asymptotic analysis cited
in the proof of Proposition \ref{prop:lich-limit-hat}.

\subsection{Sensitivity with respect to a coupling coefficient}\label{sec:sensitive}
The results of the previous sections depend in a sensitive way on a coupling constants in equations 
\eqref{CTS-reduced}.  Consider the following variation of the Einstein constraint equations:
\be
\begin{aligned}
R_{h}-\abs{K}_{h}^2+ \tr_{h} K ^2 &= 0 &\qquad\\
\div_{h} K -(1+\epsilon)\,d\,\tr_{h} K &= 0.
\end{aligned}
\ee
The case $\epsilon=0$ corresponds with the standard constraint equations.  Repeating the analysis
above for these perturbed constraint equations the analogue of equation \eqref{model-lich} is
\bel{model-lich-eps}
-2 \kappa q \dexp \psi_d'' - 2\eta^2 d^{-2q} \psi_d^{-q-1}- \kappa[(\gamma_N+\Tbeta)(1+\epsilon)+d^{-q}\mu  ]^2 \psi_d^{-q-1}
+\kappa(\Ta+\Tbeta)^2\psi_d^{q-1} = 0.
\ee
One readily shows that estimate \eqref{F-lower-sing} of Lemma \ref{lem:F-bounds} holds for
this equation, as does Lemma \ref{lem:upper-is-enough}, so long as $\epsilon>-1$.
Hence there exists a solution of the constraints for this data if and only if $\calF(d)\le 1$ for some $d>0$.

Recall that for the standard conformal method (i.e. when $\epsilon=0$), $\lim_{d\ra\infty} \calF(d) = 1$ if $\abs{\Ta}<1$. Since we are seeking solutions of $\calF(d)=1$,  it is as if there is a
solution of $\calF(d)=1$ at $d=\infty$. Adjusting $\epsilon$ affects the value of this limit.
We will show that when $\epsilon<0$, $\lim_{d\ra\infty} \calF(d)<1$, and 
the solution at $d=\infty$ becomes a true solution.  On the other hand, for $\epsilon>0$, 
$\lim_{d\ra\infty} \calF(d)>1$ and this allows for there to be no
solutions at all of $\calF(d)=1$ for sufficiently small transverse-traceless tensors.

We first show that when $\epsilon <0$, we have existence under rather general conditions,
and lose the non-existence results of Theorems \ref{thm:exceptional} and \ref{thm:critical-eta}.
\begin{proposition}
Suppose $-1<\epsilon<0$ and $\Ta\neq 1$.  If either $\mu\neq 0$ or $\eta\neq 0$
then there exists at least one solution of equation \eqref{model-lich-eps}.
\end{proposition}
\begin{proof}
Following the the arguments leading to Proposition \ref{prop:lich-limit-hat} we see that
\be
\lim_{d\ra\infty} \hatpsi{d}(0) = 
\begin{cases} 
|1+\epsilon|^{\frac{1}{q}} & \abs{\Ta}<1 \\
|1+\epsilon|^{\frac{1}{q}} |\Ta|^{-\frac{1}{q}} & \abs{\Ta}> 1.
\end{cases}
\ee
Since $|1+\epsilon| < 1$, we see that for any choice of $\Ta\neq 1$,  
$\hatpsi{d}(0)<1$ for $d$ sufficiently large.  The arguments of Section \ref{sec:critical}
can then be repeated to show that $\lim_{d\ra\infty}\psi_d(0)=\lim_{d\ra\infty}\hatpsi{d}(0)$
and hence $\psi_d(0)<1$ for $d$ sufficiently large.
Hence there exists at least one solution.
\end{proof}

Raising the value of the coupling coefficient, i.e. when $\epsilon>0$, we lose the small-TT result \ref{thm:small-TT}.
\begin{proposition}
Suppose $\epsilon>0$. If $\Ta$ is sufficiently close to $\gamma_N$, and if $\mu=0$, then there does not
exist a solution of \eqref{model-lich-eps}.
\end{proposition}
\begin{proof}
We will show that $\phi=1+\delta$ is a subsolution of \eqref{model-lich-eps} for any $d>0$ if 
$\delta>0$ is sufficiently small and  $\Ta$ is 
sufficiently close to $\gamma_N$. Having shown this we
conclude that $\calF(d)\ge 1+\delta$ for all $d>0$ and hence there are no solutions.

Note that $\phi=1+\delta$ is a subsolution (for $\mu=0$) if
\bel{delta-sub}
-2\eta^2d^{-2q}(1+\delta)^{-q-1}-(1+\epsilon)^2(\gamma_N+\Tbeta)^2(1+\delta)^{-q-1}+(\Ta+\Tbeta)^2(1+\delta)^{q-1} \le 0.
\ee
First, consider the case $\delta=0$.  We then wish to show that
\be
-2\eta^2d^{-2q}-(1+\epsilon)^2(\gamma_N+\Tbeta)^2+(\Ta+\Tbeta)^2 \le 0.
\ee
Since $\epsilon>0$, 
\be
-(1+\epsilon)^2(\gamma_N+\Tbeta)^2+(\Ta+\Tbeta)^2
\ee
is strictly negative if $\Ta=\gamma_N$. Hence the left-hand side of \eqref{delta-sub} is negative if $\delta=0$, and it
is easy to see that it remains negative if $\delta>0$ is sufficiently small.  For any such $\delta$,
we observe that this condition also holds for $\Ta$ sufficiently close to $\gamma_N$.
\end{proof}

\section{A singularly perturbed Lichnerowicz equation}\label{sec:asympt}
The most interesting results of Section \ref{sec:roughfamily} concerning 
non-existence/non-uniqueness depend on the asymptotic analysis of this section.
We consider the singularly perturbed Lichnerowicz equation
\begin{equation}\label{sing-lich}
-\epsilon^2 u_\epsilon'' - \alpha^2 u_\epsilon^{-q-1} + \beta^2 u_\epsilon^{q-1} = 0
\end{equation}
on $S^1$, which we take to be $[-\pi,\pi]$ with endpoints identified.
We assume that the functions $\alpha$ and $\beta$ are constant on the intervals
$I_-=(-\pi,0)$ and $I_+=(0,\pi)$ taking on the values $\alpha_\pm$ and $\beta_\pm$.
Proposition \ref{prop:lich-low-reg} implies that 
there exists a (unique) 
solution $u_\epsilon\in W^{2,\infty}_+(S^1)$ of \eqref{sing-lich}
so long as one of $\alpha_{\pm}\neq 0$ and one of $\beta_{\pm}\neq 0$.
By uniqueness of the solution we note that it is even about $x=\pi/2$.

\begin{figure}
\begin{center}
\includegraphics{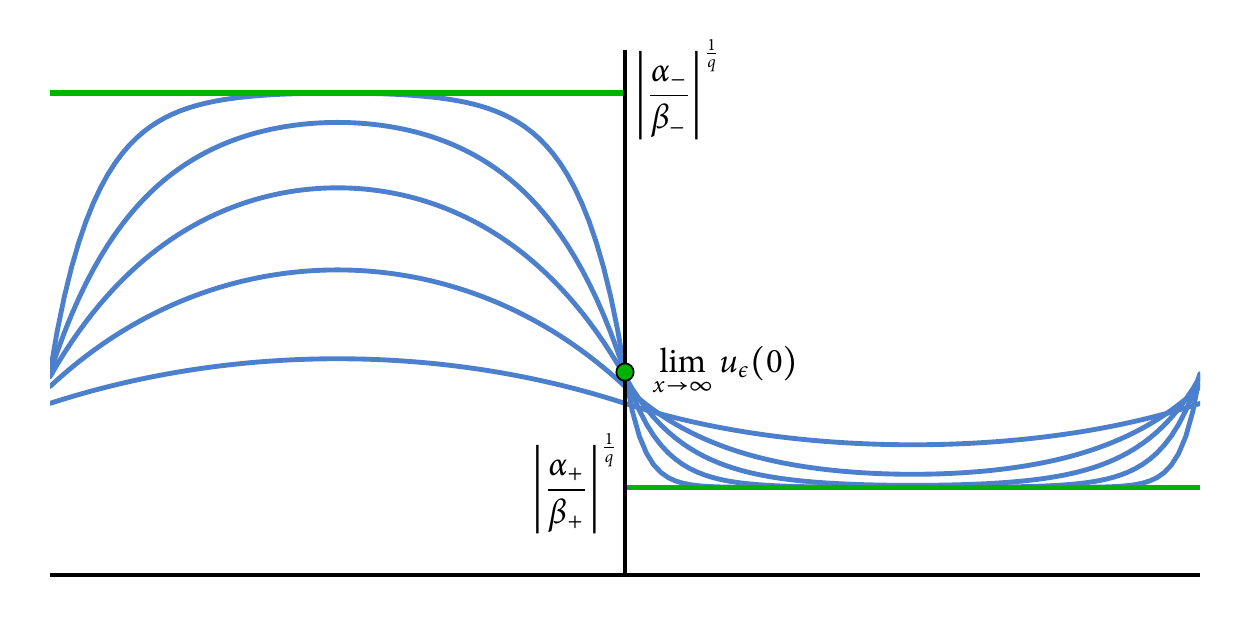}
\caption{Functions $u_\epsilon$ and their limit as $\epsilon\ra0$.}
\end{center}
\label{fig:asympt}
\end{figure}
As $\epsilon\ra 0$, equation \eqref{sing-lich} becomes an algebraic equation for $u_\epsilon$
and we expect that, away from the points of discontinuity of $\alpha$ and $\beta$,
that $u_\epsilon$ converges to the algebraic solution 
$u_0=\abs{\alpha_\pm/\beta_\pm}^{1/q}$ on $I_\pm$; see Figure \ref{fig:asympt}.
We are concerned with the behaviour of $u_\epsilon$ at the point of discontinuity, 
i.e. $\lim_{\epsilon\ra 0^+} u_\epsilon(0)$.

The principal result of this section is the following.
\begin{theorem}\label{thm:lich-limit}
Suppose that $\beta_-\neq 0$ and $\beta_+\neq 0$.  Then
\bel{lich-limit-S1}
\lim_{\epsilon\ra 0} u_\epsilon(0) = \left[\frac{\abs{\alpha_+}+\abs{\alpha_-} }{\abs{\beta_+}+\abs{\beta_-}}\right]^{\frac{1}{q}}, 
\ee
and this convergence is rapid (as defined in Definition \ref{def:rapid}).
\end{theorem}	

To obtain the limit at zero, we use a blow-up argument, guessing an asymptotic form of the solution.
We start with a boundary value problem on $[0,\infty)$.
\begin{proposition}\label{prop:Utemplate}
Let $u_0>0$.  There exists a solution on $[0,\infty)$ of 
\bel{Uode}
-U''=U^{-q-1}-U^{q-1}
\ee
satisfying $U(0)=u_0$ and $\lim_{x\ra\infty} U(x) = 1$ (with $U$ converging rapidly to
its limit at $\infty$).  Moreover, $U$ satisfies the first order equation
\bel{Uprime}
U' = \sqrt{\frac2q}\left[ U^{-q/2}- U^{q/2}\right]
\ee
and $U'(x)\ra 0$ rapidly as $x\ra\infty$.
\end{proposition}
\begin{proof}
We construct a solution by means of the method of reduction of order.

Suppose $0<u_0<1$.  Define 
\be
X(u) = \sqrt{\frac q2}\int_{u_0}^u \frac{v^{q/2}}{1-v^q}\; dv
\ee	
for $u_0<u<1$.  Note that $\lim_{u\ra 1^-} X(u) = \infty$.  Moreover,
$ X'(u) = \sqrt{\frac2q}\; u^{q/2}/(1-u^q) >0$.  Hence $X$ has an increasing
inverse function
\be
U:[0,\infty) \ra [u_0,1)
\ee
satisfying $U(0)=u_0$ and $\lim_{x\ra\infty} U(x) = 1$.  Moreover,
\bel{Uprime2}
U'(x) = \frac{1}{X'(U(x))} = \sqrt{\frac2q}\left[ U(x)^{-q/2}-U(x)^{q/2}\right].
\ee
An easy computation involving the chain rule and equation \eqref{Uprime} now shows that $U$
satisfies the ODE \eqref{Uode} and hence $U$ is the function we seek.

If $u_0>1$ one shows similarly that the inverse function of 
\be
X(u) = \sqrt{\frac2q}\int_{u}^{u_0} \frac{v^{q/2}}{v^q-1}\; dv
\ee 
defined on $(1,u_0]$ is the desired function.  When $u_0=1$, then $U(x)\equiv 1$ is the solution.

To show the rapid convergence at infinity we focus on the case $0<u_0<1$.  Let $W=1-U$,
so $W>0$ and $\lim_{x\ra\infty}W(x)=0$.  Now 
\be
W' = \sqrt\frac{2}{q}\left[ (1-W)^{q/2}-(1-W)^{-q/2}\right] = H(W) W
\ee
where $H$ is a continuous function near $0$ and 
\be
H(0) = \left.\frac{d}{dW}\right|_{W=0} \sqrt\frac{2}{q}\left[ (1-W)^{q/2}-(1-W)^{-q/2}\right] = -\sqrt{2q}.
\ee
Since $W(x)\ra 0$ as $x\ra\infty$, there exists $x_0$ so that if $x\ge x_0$, 
\be
H(W(x)) < -\sqrt{q}
\ee
Hence
\be
W'\le -\sqrt{q}\, W
\ee
for $x\ge x_0$ and by Gronwall's inequality
\be
W(x) \le W(x_0)\exp(-\sqrt{q}\, x).
\ee
Since $W\ge 0$ also, we conclude that $W$ converges rapidly to $0$ and $U$ converges rapidly to $1$.

The rapid convergence when $u_0>1$ is proved similarly, while the result is trivial if $u_0=1$.  
Finally, we note that
the rapid convergence of $U'$ to $0$ at infinity follows from the rapid convergence of $U$ to $1$ at 
infinity and equation \eqref{Uprime}.
\end{proof}

We now turn to a boundary value problem on $\Reals$ with piecewise constant coefficients.
Consider
\be\label{lich-R}
-v''-\alpha^2 v^{-q-1}+\beta^2 v^{q-1} = 0
\ee
on $\Reals$ where $\alpha$ and $\beta$ are equal to the constants $\alpha_\pm$ and $\beta_\pm$
on the intervals $(0,\infty)$ and $(-\infty,0)$.
\begin{proposition}\label{prop:lich-R}
Suppose $\beta_\pm\neq 0$. Let $L_\pm=\abs{\alpha_\pm/\beta_\pm}^{1/q}$.
There exists a solution  $v\in W^{2,\infty}_{\rm loc}(\Reals)$ of \eqref{lich-R} satisfying
\be
\lim_{x\ra\pm\infty} v(x) = L_\pm.
\ee
Moreover, $v$ converges rapidly to its limits at $\pm\infty$, $v'$ converges rapidly to $0$ at 
$\pm\infty$, and
\bel{v0}
v(0) = \left[\frac{\abs{\alpha_+}+\abs{\alpha_-}}{\abs{\beta_+}+\abs{\beta_-}}\right]^{\frac{1}{q}}.
\ee
\end{proposition}
\begin{proof}
Let $\omega_{\pm}= \left[\alpha_\pm^{q+1}\beta_\pm^{q-1}\right]^{\frac{1}{2q}}$.  
Given any $c>0$ we define
\be
v_c = \begin{cases} L_+ U_+(\omega_+ x) & x>0\\
L_- U_-(-\omega_- x) & x<0
\end{cases}
\ee
where $U_\pm$ is the solution of \eqref{Uode} provided by Proposition \ref{prop:Utemplate} satisfying
$U_\pm(0)= \frac{c}{L_\pm}$ and $\lim_{x\ra\infty} U_\pm(x) = 1$.
Then $v_c$ is continuous, satisfies the differential equation
\eqref{lich-R} on $(0,\infty)$ and $(-\infty,0)$,
and has the correct limiting behaviour at $\pm\infty$.
If for some $c$, $v_c$ is differentiable at $0$, then $v_c$ will be a weak solution on $\Reals$ and
by elliptic regularity the desired strong solution.

From Proposition \ref{prop:Utemplate} we have
\be
v_c'(0+) = L_+\omega_+U_+'(0) = L_+\omega_+\sqrt{\frac{2}{q}}\left[ \left(\frac{c}{L_+}\right)^{-q/2}-
\left(\frac{c}{L_+}\right)^{q/2}\right]
\ee
and similarly
\be
v_c'(0-) = -L_-\omega_-\sqrt{\frac{2}{q}}\left[ \left(\frac{c}{L_-}\right)^{-q/2}-
\left(\frac{c}{L_-}\right)^{q/2}\right].
\ee
Setting these quantities equal we obtain
\be
\left[ L_+\omega_+L_+^{-q/2}+L_-\omega_-L_-^{-q/2}\right]c^q = L_+\omega_+L_+^{q/2}+L_-\omega_-L_-^{q/2}.
\ee
From the definitions of $L_\pm$ and $\omega_\pm$ we have the identities
\be
L_\pm^2\omega_\pm^2 = \alpha_+\pm^2L_\pm^{-q} = \beta_\pm^2L_\pm^q
\ee
and hence
\be
c^q = \frac{\abs{\alpha_+}+\abs{\alpha_-}}{\abs{\beta_+}+\abs{\beta_-}}.
\ee
With this choice of $c$ we obtain a solution of \eqref{lich-R} satisfying equation \eqref{v0}.
\end{proof}

Using the function found in Proposition \ref{prop:lich-R} we can construct approximate solutions
of the differential equation \eqref{sing-lich}.  Our strategy for proving Theorem \ref{thm:lich-limit} will be 
to show that these approximate solutions improve as $\epsilon\ra 0$ and can be corrected using
Newton's method to obtain solutions satisfying the limit \eqref{lich-limit-S1}.

We form the approximate solutions first on $[-\pi/2, \pi/2]$, defining
\be
w_\epsilon(x) = v(x/\epsilon)  + h_\epsilon(x)
\ee
where $h_\epsilon$ will be a small correction term. We will pick $h_\epsilon$ 
so that $w_\epsilon'(\pm\pi/2)=0$ and hence can we can extend $w_\epsilon$
to be defined on $S^1$ by declaring it to be even about $x=\pi/2$.  

To define the correction term, we first let
\be
\zeta(x)=\begin{cases} \frac{1}{\pi}x^2 & 0\le x \le \pi/2 \\
                        0 & -\pi/2 < x \le 0
\end{cases}
\ee
and note that $\zeta'(\pi/2)=1$.  Let
\be
h_\epsilon(x)=-d_{\epsilon,+}\zeta(x) -d_{\epsilon,-}\zeta(-x)
\ee
where
\be
d_{\epsilon,\pm}=\frac{1}{\epsilon} v'(\pm \pi/(2\epsilon)).
\ee
With this choice of $h_\epsilon$, $w_\epsilon'(\pm\pi/2)=0$.

For $p>1$ we define the nonlinear Lichnerowicz operator 
$\calN_\epsilon: W^{2,p}(S^1)\ra L^p(S^1)$ by
\be
\calN_\epsilon(w) = -\epsilon^2 w'' - \alpha^2 w^{-q-1} + \beta^2w^{q-1}.
\ee
The error $\calE_\epsilon = \calN_\epsilon(w_\epsilon(x))$ is even about $x=\pi/2$ and
one readily computes that on $[-\pi/2,\pi/2]$,
\be
\calE_\epsilon=
-\alpha^2 \left[ (v(x/\epsilon)+h_\epsilon(x))^{-q-1}-v(x/\epsilon)^{-q-1} \right]
+\beta^2 \left[ (v(x/\epsilon)+h_\epsilon(x))^{q-1}-v(x/\epsilon)^{q-1}\right]
+\frac{2\epsilon^2}{\pi}\left[d_+\chi_+ + d_-\chi_-\right]
\ee
where $\chi_\pm$ are the characteristic functions of $(0,\pi)$ and $(-\pi,0)$ respectively.
\begin{lemma}\label{lem:err-to-zero}
	\be ||\calE_\epsilon||_{L^\infty(S^1)}\ra 0 \ee
	rapidly as $\epsilon\ra 0$.
\end{lemma}
\begin{proof}
From Proposition \ref{prop:lich-R} we know that $v'(x)\ra 0$ rapidly as $x\ra\infty$.
Consequently the constants $d_{\epsilon,\pm}$ converge rapidly to zero as $\epsilon \ra 0$.
Moreover, $d_{\epsilon,+}\chi+$ and $d_{\epsilon,-}\chi_-$ converge rapidly to
$0$ in $L^\infty(S^1)$.

Consider $F(v)=v^{-q-1}$.  Then $F(v+h)-F(v) = (-q-1) \int_0^1 (v+th)^{-q-2}dt\; h $ and
therefore
\be
\abs{ F(v+h)-F(v) }  \le (q+1)\max_{t\in [0,1]} (v+th)^{-q-2} |h|.
\ee
Now $v_\epsilon(x) \ge \min( L_+,L_-) > 0$ and $h_\epsilon$ converges rapidly to $0$ in 
$L^\infty([-\pi/2,\pi/2])$.  So there is an $m$ such that
\be
v_\epsilon+t h_\epsilon \ge m > 0.
\ee
for all $t\in[0,1]$ and all $\epsilon$ sufficiently small.
It follows that
\be
|| (v_\epsilon+h_\epsilon)^{-q-1}-v_\epsilon^{-q-1}||_{L^\infty([-\pi/2,\pi/2])} \le (q+1) m^{-q-2} ||h_\epsilon||_{L^\infty([-\pi/2,\pi/2])}
\ee
for $\epsilon$ sufficiently small.  From the rapid convergence of $h_\epsilon$ to zero we conclude that
\be
\alpha^2\left[(v_\epsilon+h_\epsilon)^{-q-1}-v_\epsilon^{-q-1}\right] \ra 0
\ee
rapidly in $L^\infty([-\pi/2,\pi/2])$ as $\epsilon\ra 0$.  A similar argument establishes
\be
\beta^2\left[(v_\epsilon+h_\epsilon)^{q-1}-v_\epsilon^{q-1}\right] \ra 0
\ee
rapidly as $\epsilon \ra 0$.  

We have considered all terms of $\calE_\epsilon$ and conclude that 
\be
||\calE_\epsilon||_{L^{\infty}([-\pi/2,\pi/2])}\ra 0
\ee
rapidly as $\epsilon\ra 0$. Since $\calE_\epsilon$ is even about $x=\pi/2$, we have the same
convergence in $L^{\infty}(S^1)$.
\end{proof}

For constants $0<m<M$ and $p>1$ we define the slab $S_{m,M}^p = \{ u\in W^{2,p}(S^1) : m\le u\le M\}$.
\begin{lemma}\label{lem:fp}
For $u\in W^{2,p}_+(S^1)$, let $F_r(u) = u^r$. There exists a constant $K(m,M,r)$ such that
\bel{Fr-is-Lip}
||F_r(u)-F_r(v)||_{L^p(S^1)} \le K(m,M,r) ||u-v||_{L^p(S^1)}
\ee
for all $u,v\in S^p_{m,M}$.

Let $L_{u,r}:W^{2,p}\ra L^p$ be the linear function
\be
L_{u,r} v = F_r(u) v.
\ee
The map $u\mapsto L_{u,r}$ is Lipschitz continuous on $S^p_{m,M}$.
\end{lemma}
\begin{proof}
Note that if $m\le x,y \le M$ then
\be
x^r-y^r =\int_0^1 r ((1-t) x + ty )^{r-1} \; dt\, (x-y)
\ee
and hence
\be
|x^r-y^r| \le r (m^{r-1}+M^{r-1}) \abs{x-y}.
\ee
Consequently
\be
||F_r(u)-F_r(v)||_{L^p(S^1)} \le r (m^{r-1}+M^{r-1}) ||u-v||_{L^p(S^1)}.
\ee
Inequality \eqref{Fr-is-Lip} now follows setting $K=r(m^{r-1}+M^{r-1})$.

If $u_1,u_2\in S_{m,M}^p$ and $v\in W^{2,p}$, then
\be
\begin{aligned}
|| L_{u_1,r}v - L_{u_2,r}v||_{L^p(S^1)} &= ||(F_r(u_1)-F_r(u_2))v||_{L^p(S^1)} \\
&\le ||F_r(u_1)-F_r(u_2)||_{L^p(S^1)} || v||_{L^\infty(S^1)}\\
&\le K(m,M,r)||u_1-u_2||_{L^p(S^1)} ||v||_{W^{2,p}(S^1)}\\
&\le K(m,M,r)||u_1-u_2||_{W^{2,p}(S^1)} ||v||_{W^{2,p}(S^1)}.
\end{aligned}
\ee
Hence $||L_{u_1,r}-L_{u_2,r}|| \le K(m,M,r)||u_1-u_2||_{W^{2,p}(S^1)}$ which
establishes the Lipschitz continuity.
\end{proof}

One readily shows that the linearization of 
$\calN_\epsilon$ at $w$ is the operator $\calN'_{\epsilon}[w]$ defined by
\be
\calN_\epsilon'[w] h = -\epsilon^2 h'' + [(q+1)\alpha^2 w^{-q-2} +(q-1)\beta^2 w^{q-2}] h.
\ee
As an immediate consequence of Lemma \ref{lem:fp} we see that $\calN_\epsilon'$ is Lipschitz
continuous.
\begin{corollary}\label{cor:Lip}
Suppose $0<m<M$.  There exists a constant $C(m,M)$ such that for  all $v,w\in S_{m,M}^p$, 
\be
||\calN_\epsilon'[v]-\calN_\epsilon'[w]||_{L(W^{2,p}(S^1),L^p(S^1))} < C(m,M) ||v-w||_{W^{2,p}(S^1)}.
\ee
\end{corollary}

Our application of Newton's method requires an estimate of the size of ${\calN'_\epsilon}^{-1}$
as $\epsilon\ra 0$, which we obtain next.
\begin{proposition}  \label{prop:inv}
Let $V\in L^\infty(S^1)$ and consider the operator
\be
\calL_\epsilon=-\epsilon^2 \Lap + V
\ee
as a map from $W^{2,p}(S^1)$ to $L^p(S^1)$, where $p>1$.  Suppose there is a constant $m$ such
that $V\ge m > 0$.  Then $\calL_\epsilon$
is continuously invertible. Moreover, there is a constant $C$ such that if $\epsilon$ is
sufficiently small,
\bel{L-est}
|| \calL_{\epsilon}^{-1} || \le C\epsilon^{-4}.
\ee
\end{proposition}
\begin{proof}
The fact that $\calL_\epsilon$ is continuously invertible follows from standard elliptic
theory and the positivity of $V$. We turn our attention to obtaining the estimate \eqref{L-est}.

Let $S^1_r$ denote the circle of radius $r$, and let $i_r:S^1_r \ra S^1$ be the
natural diffeomorphism.  For a function $u$ define on $S^1$ let $u_r=u\circ i_r$. Suppose 
\be
-\epsilon^2 u'' + V u  = f
\ee
on $S^1$. Letting $r=1/\epsilon$ we then have
\be
-u_r'' + V_r u_r  = f_r.
\ee
Let $I$ be an interval of length $1$ in $S^1_r$ and let $I'$ be the interval
of length $1/2$ at the center of $I$.  From interior $L^p$ estimates (\cite{epde} Theorem 9.11)
we have
\be
||u_r||_{W^{2,p}(I')} \le C_1 \left[ ||f_r||_{L^p(I)} + ||u_r||_{L^p(I)}\right]
\ee
where $C_1$ depends on $||V||_\infty$ but does not depend on $I$ or $r$.
Averaging these interior estimates over all intervals $I$ in $S^1_r$ we obtain
\be
||u_r||_{W^{2,p}(S^1_r)} \le C_2 \left[ ||f_r||_{L^p(S^1_r)} + ||u_r||_{L^p(S^1_r)}\right].
\ee
where $C_2$ (and all subsequent constants $C_k$) is independent of $r$ (and $\epsilon$).
One readily verifies that for any function $w$ on $S^1$,
\be
||\nabla^k w||_{L^p(S^1)} = r^{k-\frac{1}{p}}||\nabla^k w_r||_{L^p(S^1_r)}.
\ee
Assuming that $r>1$ (i.e. $\epsilon <1 $) it then follows that
\be
|| u||_{W^{2,p}(S^1)} \le C_3 r^{2-\frac1p} || u_r||_{W^{2,p}(S^1_r)}.
\ee
and therefore
\bel{w2p-control-1}
\begin{aligned}
||u||_{W^{2,p}(S^1)} &\le r^{2-\frac{1}{p}} C_2C_3\left[ ||f_r||_{L^p(S^1_r)} + 
||u_r||_{L^p(S^1_r)}\right]\\
&= \epsilon^{-2} C_2C_3\left[ ||f||_{L^p(S^1)} + ||u||_{L^p(S^1)}\right].
\end{aligned}
\ee
By Sobolev embedding in $S^1$ we have for some constant $C_4$ 
\bel{sob}
||u||_{L^p(S^1)} \le C_4 ||u||_{W^{1,2}(S^1)}.
\ee
Suppose $\epsilon < \sqrt{m}$.  Then
\be
\begin{aligned}
||u||_{W^{1,2}(S^1)}^2 &= \int_{S^1} \abs{\nabla u}^2 + u^2 \\
                &\le \max(\epsilon^{-2},m^{-1})  \int_{S^1} \epsilon^2 \abs{\nabla u}^2 + V u^2 \\
                & = \max(\epsilon^{-2},m^{-1})  \int_{S^1} fu \\
                & \le \epsilon^{-2} ||f||_{L^p(S^1)} ||u||_{L^{p'}(S^1)}
\end{aligned}
\ee
where $p'$ is the conjugate exponent to $p$. By Sobolev embedding again we have
$||u||_{L^{p'}(S^1)}\le C_5 ||\phi||_{W^{1,2}(S^1)}$ and hence
\bel{w12-control}
||u||_{W^{1,2}(S^1)} \le C_5 \epsilon^{-2} ||f||_{L^p(S^1)}
\ee
Combining inequalities \eqref{w2p-control-1}, \eqref{sob}, and \eqref{w12-control} we obtain
\be
||u||_{W^{2,p}(S^1)} \le C_2C_3\left[\epsilon^{-2} +C_4C_5\epsilon^{-4}\right] ||f||_{L^p(S^1)}.
\ee
if $\epsilon<\min(1,\sqrt{m})$. Since $\epsilon < 1$, this establishes
inequality \eqref{L-est} with $C = C_2C_3(1+C_4C_5)$.
\end{proof}

We are now in a position to prove our main result of the section.
\begin{proof}[Theorem \ref{thm:lich-limit}]
The proof involves Newton's method, and we briefly recall the  required hypotheses here (\cite{NFA}). 
Let $X$ and $Y$ be Banach spaces, $x\in X$, $r>0$.  Let $\calN:B_r(x)\ra Y$
be a differentiable map with Lipschitz continuous derivative, i.e. there exists $k>0$
such that 
\be
|| \calN'[x_1]-\calN'[x_2]||_{L(X,Y)} \le k ||x_1-x_2||_X
\ee
for all $x_1,x_2\in B_r(x)$. Suppose $x$ is a point where $\calN'[x]$ has a continuous inverse.
Let $c_1=||\calN(x)||$ and let $c_2=||\calN'[x]^{-1}||$. If
$2kc_1^2c_2<1$ and $2c_1c_2<r$, then there exists a solution of $\calN(u) = 0$ satisfying
$||u-x||_X \le 2c_1c_2$.

We apply this method to the operators $\calN_\epsilon$.  Let $m=\inf v$ and $M=\sup v$ where
$v$ is the asymptotic solution found in Proposition \ref{prop:lich-R}.  Taking $\epsilon$
sufficiently small we can assure that $m/2 < w_\epsilon <2 M$.  
By the imbedding of $W^{2,p}(S^1)$ into $C^0(S^1)$ we can find an $r$ such that if $m/2<w<2M$ and
$u\in B_r(w)$, then $m/3<u<3M$.  Let $k$ be the Lipschitz constant for $\calN_\epsilon'$ on
$S_{m/3,3M}^p$ obtained in Corollary \ref{cor:Lip}.  So for $\epsilon$ sufficiently small,
$\calN'_\epsilon$ is Lipschitz continuous with constant $k$ on $B_r(w_\epsilon)$.

Let $c_1(\epsilon) = ||\calN_\epsilon(w_\epsilon)||_{L^p}$ and let $c_2(\epsilon) = ||{\calN'_\epsilon}^{-1}||$.  By Lemma \ref{lem:err-to-zero}
$c_1(\epsilon)$ converges rapidly to zero, while by Proposition \ref{prop:inv}, 
$c_2(\epsilon)$ is $O(\epsilon^{-4})$.  Hence $2kc_1c_2^2$ and $2c_1c_2$ converge
rapidly to zero, and for $\epsilon$ sufficiently small we obtain a solution of
$\calN_\epsilon(u_\epsilon)=0$ with $||u_\epsilon-w_\epsilon||_{W^{2,p}(S^1)}< 2c_1c_2$. By 
the continuous imbedding of $W^{2,p}(S^1)$ into $C^0(S^1)$ we have in particular that
$u_\epsilon(0)$ converges rapidly to $w_\epsilon(0)=v(0)$ as $\epsilon\ra 0$.  
Since $u_\epsilon$ is the unique solution of \eqref{sing-lich}, we have proved the result.
\end{proof}

\section{Conclusion}
By working with a concrete model problem, we have observed a number of new phenomena
for the vacuum conformal and CTS methods.  For certain conformal
data violating both a small-TT and a near-CMC condition we have shown that there 
cannot be a unique solution: there will either be no solutions or more than one.
For other small-TT data violating a near-CMC we have shown that there are multiple solutions.  
We have also found 
existence of certain solutions under a very weak near-CMC hypotheses ($\tau$ has constant sign),
dependence of the solution theory on the lapse function or conformal class
representative, and extreme sensitivity of the solution theory with respect to a coupling constant 
in the Einstein constraint equations.

This work was motivated by the following questions that arise from the Yamabe-positive
small-TT existence theorems of \cite{Holstetal08} and \cite{Maxwell09}:
\begin{enumerate}
	\item Is the small-TT hypothesis required to ensure existence for arbitrary mean curvatures?
	\item Are small-TT solutions necessarily unique?
	\item Can the Yamabe-positive restriction be relaxed?
\end{enumerate}
Our examples were obtained using a Yamabe-null background metric, and therefore do not
directly address questions 1) and 2).  The answers to these questions in the Yamabe-null
case, however, are that the small-TT hypothesis is necessary 
(at least for the existence of symmetric solutions for symmetric data), and that small-TT solutions
need not be unique. Moreover, our coefficient sensitivity results also suggest that if 
it is possible to extend the existence results of \cite{Holstetal08} and \cite{Maxwell09} to
Yamabe-null manifolds, the proof will be difficult.

These negative results suggest that the conformal and CTS methods do not lead to
a good parameterization scheme for solutions of the Einstein constraint equations.  
Since the conformal method, in its CMC formulation, is so successful, one is lead to
wonder if there is some other generalization of it that does lead to a parameterization.
This remains to be seen, and the model problem developed here could provide a useful
test case for investigating possible alternatives.

\subsection*{Acknowledgement}
I would like to thank Daniel Pollack and Jim Isenberg for useful
comments and discussions.

\ifjournal
\bibliographystyle{mrl}
\else
\bibliographystyle{amsalpha-abbrv}
\fi
\bibliography{mrabbrev,noncmc,rc,nonuniq}

\appendix
\section{The Lichnerowicz equation}\label{sec:lich}
We give the proof here of Proposition \ref{prop:lich-low-reg}
concerning solutions of the differential equation 
\bel{lich-generic-app}
-u'' -\alpha^2 u^{-q-1} + \beta^2 u^{q-1} = 0
\ee
on $S^1$.
\bgroup
\def\thetheorem{\ref{prop:lich-low-reg}}
\begin{proposition}
Suppose $\alpha$ and $\beta$ in equation \eqref{lich-generic-app}
belong to $L^\infty(S^1)$ and that $\alpha\not\equiv 0$ and $\beta\not\equiv 0$. 
Let $p>1$.
\begin{enumerate}
\item There exists a unique solution $u \in W^{2,p}_+(S^1)$, and $u\in W^{2,\infty}(S^1)$.
\item If $w\in W^{2,\infty}_+(S^1)$ is a subsolution of \eqref{lich-generic}, (i.e. $-w''-\alpha^2w^{-q-1}+\beta^2w^{q-1}\le 0$) then $w\le u$.
\item If $v\in W^{2,\infty}_+(S^1)$ is a supersolution of \eqref{lich-generic}, (i.e. $-v''-\alpha^2v^{-q-1}+\beta^2v^{q-1}\ge 0$) then $v\ge u$.
\item The solution $u\in W^{2,p}_+$ depends continuously on $(\alpha,\beta)\in L^\infty\times L^\infty$.
\end{enumerate}
\end{proposition}
\egroup
\begin{proof}
We consider the differential equation \eqref{lich-generic-app} to hold on $S^n$ rather than $S^1$ so as to be 
able to cite existing work (recall that $n$ is related to $q$ by $q=2n/(n-2)$). That is, we consider
\bel{lich-high}
-\Lap u -\alpha^2 u^{-q-1}+\beta^2u^{q-1} = 0
\ee
on $(S^n,g)$, where $\alpha$ and $\beta$ depend only on $x^n$.  

Since $\alpha^2\not\equiv 0$ and $\beta^2\not\equiv 0$,
\cite{cb-low-reg} Theorem 4.10 and Corollary 4.11 imply that there exists a 
positive solution in $W^{2,p}$ for $p>n/2$.
Uniqueness of this solution follows from \cite{cb-low-reg} Theorem 4.9.  
From uniqueness we know that $u$ is a function of $x^n$ alone (otherwise
translation along $x^k$ with $1\le k\le n-1-$ would yield a different solution).
But then equation \eqref{lich-high} reduces to equation \eqref{lich-generic-app}.  
This establishes the existence in Part 1 for $p>1$ and uniqueness for $p>n/2$.
By an easy bootstrap, if $u$ is a solution of \eqref{lich-generic-app} in $W^{2,p}_+$ for
some  $1<p\le n/2$, then in fact $u\in W^{2,p}_+$ for all $p>1$ (and indeed also
$W^{2,\infty}$) and hence it is the unique solution.

Suppose $u_- \in W^{2,\infty}_+$ is a subsolution of \eqref{lich-generic-app}.  Then it is also
a subsolution of \eqref{lich-high}.  Let $u$ be the positive
solution of \eqref{lich-generic-app}. Arguing as in \cite{Maxwell09} Lemma 2 it follows
that $M u$ is a supersolution for any $M>1$.  Pick $M$ so that $u_-\le M u$.
Proposition 8.2 of \cite{cb-low-reg} implies there is a solution $v$ of \eqref{lich-high}
such that $u_-\le v \le M u$. By uniqueness of the solution it follows that $v=u$.
Hence $u_-\le u$ and we have proved Part 2.  Part 3 is proved similarly.

To show continuity, we use the Implicit Function Theorem.
Consider the map $\calN:W^{2,p}_+\times (L^\infty\times L^\infty)\mapsto L^p $ taking
\be
(u,\alpha,\beta) \mapsto -u''-\alpha^2 u^{-q-1} +\beta^2 u^{q-1}.
\ee
This map is evidently continuous (since $W^{2,p}$ is an algebra). One readily
shows that its Fr\'echet derivative at $(u,\alpha,\beta)$ with respect to $u$
in the direction $h$ is
\be
\calN'[u,\alpha,\beta] h = -h'' +[(q+1)\alpha^2 u^{-q-2} + (q-1)\beta^2 u^{q-2}]h
\ee
The continuity of the map $(u,\alpha,\beta)\mapsto \calN'[u,\alpha,\beta]$
follows from the fact that $W^{2,p}(S^1)$ is an algebra continuously embedded
in $C^0(S^1)$ along with Lemma \ref{lem:fp}.  
Since $\alpha\not\equiv 0$ and $\beta\not\equiv 0$ the potential
$V = [(q+1)\alpha^2 u^{-q-2} + (q-1)\beta^2 u^{q-2}]$ is not identically zero. 
By \cite{cb-low-reg} Theorem 7.7, $-\Lap + V:W^{2,p}\ra L^p$ 
is an isomorphism.  The Implicit Function Theorem (see, e.g. \cite{NFA} Theorem 4.1)
then implies that if $u_0$ is a solution for data $(\alpha_0,\beta_0)$, there
is a continuous map defined near $(\alpha_0,\beta_0)$ taking $(\alpha,\beta)$
to the corresponding solution of \eqref{lich-generic-app}.  This establishes Part 4.
\end{proof}

We remark that the hypothesis $u_\pm \in W^{2,\infty}$ in Parts 2 and 3 can be 
weakened; we make it only for convenience so as to be able to apply Proposition 8.2 
of \cite{cb-low-reg} in a straightforward way. In our applications in Section \ref{sec:roughfamily}, 
the sub- and supersolutions are either constants or the sum of a constant and 
an element of $W^{2,\infty}$.

\section{Theory for even conformal data}\label{sec:even}

In this section we sketch how, despite the presence of a conformal Killing field,
existing techniques for the conformal method can be adapted to the model problem 
\eqref{CTS-reduced} if the conformal data satisfy an evenness hypothesis. For 
simplicity, we assume all data in this section are smooth, and we focus
on the standard conformal method (i.e. $N=1/2$).  The coupled system 
to solve is 
\begin{equation}\label{CTS-reduced-2N}
\begin{aligned}
-2\kappa q\;\phi'' - 2\eta^2\phi^{-q-1}-\kappa(\mu+w')^2\phi^{-q-1}+\kappa \tau^2\phi^{q-1} &= 0\\
w'' &= \phi^q \tau'.
\end{aligned}\end{equation} 

From Theorem \ref{thm:cmc} and dimensional reduction we have
the following result for the Lichnerowicz equation
\bel{lich-gen-even}
-\phi'' - \alpha^2 \phi^{-q-1} + \beta^2\phi^{q-1} = 0
\ee
on $S^1$.
\begin{proposition}\label{prop:basic-lich} Suppose $\alpha$ and $\beta$ belong to $C^\infty(S^1)$.  
There exists a smooth positive solution $\phi$
of \eqref{lich-gen-even} if and only if 
\begin{enumerate}
	\item $\alpha\not\equiv 0$ and $\beta\not\equiv 0$ or
	\item $\alpha\equiv 0$ and $\beta\equiv0$.
\end{enumerate}
The solution in case 1 is unique. In case 2 the solutions are the positive constants.
\end{proposition}

For the momentum constraint we consider
\bel{momentum-gen}
w'' = f
\ee
on $S^1$.
By standard elliptic theory we have the following result.
\begin{proposition}\label{prop:basic-mom}
Suppose $f\in C^\infty(S^1)$. There exists a solution $x\in C^\infty(S^1)$
of \eqref{momentum-gen} if and only if
\be
\int_{S^1} f = 0.
\ee
Any two solutions of \eqref{momentum-gen} differ by an additive constant.
\end{proposition}

Recall that we are working with functions on $S^1$ with domain of definition $[-\pi,\pi]$.
We say that a function $f$ on $S^1$ is even or odd  if $f(-x)=f(x)$ or $f(-x)=-f(x)$ for all $x\in [-\pi,\pi]$.
Subscripts $e$ and $o$ denote subspaces of even and odd functions.  Using the 
uniqueness results of Propositions \ref{prop:basic-lich} and \ref{prop:basic-mom}
we have the following easy corollaries.
\begin{corollary}\label{cor:lich-even} Suppose $\alpha$ and $\beta$ are in $C^\infty_e(S^1)$.
If Condition 1 or 2 of Proposition \ref{prop:basic-lich} holds, then the solution
$\phi$ of \eqref{lich-gen-even} belongs to $C^\infty_e(S^1)$.
\end{corollary}
\begin{corollary}\label{cor:mom-even} Suppose $f\in C^\infty_o(S^1)$ and $N\in C^\infty_e(S^1)$.
Then there exists a unique solution $w\in C^\infty_o(S^1)$ of \eqref{momentum-gen} satisfying $w(0)=0$.
Any other solution of \eqref{momentum-gen} can be written as $a+w$ where $a$ is constant.
\end{corollary}

Assume $\eta$, $\tau\in C^\infty_e(S^1)$ and $\mu$ is constant.
We define a map $\mathcal N:C^\infty_e(S^1)\ra C^\infty_e(S^1)$ as follows.  
Let $\phi\in C^\infty_e(S^1)$.  Then $\phi^q\tau'$ is odd and hence there exists a unique
function $w\in C^\infty_o(S^1)$ solving 
\be
w'' = \phi^q\tau'.
\ee
Let $\alpha = \frac{1}{2\kappa q}[2\eta^2+\kappa(\mu+1/(2N)w')^2]$ and $\beta=\frac{1}{2q}\tau^2$, so 
$\alpha$ and $\beta$ belong to $C^\infty_e(S^1)$.  Finally, define $\mathcal N(\phi)$ to be the 
solution of \eqref{lich-gen-even} for this choice of $\alpha$ and $\beta$. The existence
of a smooth solution of \eqref{conf-reduced} is equivalent to the existence of a fixed point of 
$\mathcal N$. 

By assuming that $\eta$ and $\tau$ are even, we have ensured that $\mathcal N$ is 
well defined and thus avoiding the trouble with conformal Killing fields.
The existence theory of \cite{Maxwell09} for the standard conformal method 
now proceeds without change and we have the following  generalization of Theorem \ref{thm:non-cmc}.
\begin{theorem}  Suppose $N\in C^\infty_e(S^1)$, $\eta,\tau\in C^\infty_e(S^1)$ and $\mu\in\Reals$.
Suppose further that $\tau\not\equiv 0$ and that either $\eta\not\equiv 0$ or $\mu\neq 0$.
If there exists a global upper barrier for $(\eta,\mu,\tau)$, then there exists a solution
$(\phi,w)\in C^\infty_e \times C^\infty_o$ of \eqref{conf-param}.
\end{theorem}

Recall that a global upper barrier is defined as follows. Given a smooth even
positive function $\phi$, let $w_\phi$ be an odd solution of
\bel{wphi}
w_\phi'' = \phi^q\tau'.
\ee
Then $w_\phi'$ is uniquely defined.  We say that a smooth positive even function $\Phi$ is 
a global upper barrier if for all smooth even functions $\phi$ satisfying $0<\phi\le \Phi$,
then
\be
-2\kappa q\;\Phi'' - 2\eta^2\Phi^{-q-1}-\kappa(\mu+w_\phi')^2\Phi^{-q-1}+\kappa \tau^2\Phi^{q-1} \ge 0.
\ee
Following \cite{isenberg-omurchadha-noncmc} and \cite{ICA08} one readily shows that
there is a constant global upper barrier if
\bel{near-cmc-2}
\frac{\max\abs{\nabla\tau}}{\min\abs{\tau}}\quad\text{is sufficiently small.}
\ee

To conclude this section, we show how we can use such near-CMC data to construct
data that violate both the small-TT condition and the near-CMC condition \eqref{near-cmc-2} arbitrarily.
To see this we `double' the frequency of the mean curvature:
if $f$ is a periodic function with period $2\pi$, let $f^{[k]}(x) = f(2^k x)$.
\begin{proposition}\label{prop:doubling}
Suppose $\tau$ satisfies the near-CMC condition \eqref{near-cmc-2}, and $\eta$ and $\mu$ are constant.  
Then for any $k\in\Nats$ there
exists a solution of \eqref{conf-param} $(\eta,\mu,\tau^{[k]})$ so long as 
one of $\eta$ or $\mu$ is non-zero.
\end{proposition}
\begin{proof}
Let $k\in \Nats$.
Since $\tau$ is near-CMC, there exists a solution $(\phi,w)$ of \eqref{conf-param}
for conformal data $(2^{-nk}\eta, 2^{-nk}\mu,\tau)$.
One verifies then that $(2^{\frac{nk}{q}}\phi^{[k]},2^{(n-1)k}w^{[k]})$
is a solution for conformal data $(\eta, \mu, \tau^{[k]})$.
\end{proof}
By taking $k$ sufficiently large, we can make the ratio 
$\frac{\max |\nabla\tau^{[k]}|}{\min \abs{\tau^{[k]}}}$ as large as we please.
For each of these mean curvatures,
we can solve \eqref{conf-param} for certain arbitrarily large TT-tensors.  This result seems to suggest that
large relative gradients of $\tau$ are not, by themselves, a source of trouble. 

The kind of near-CMC violation described above introduces large gradients without
affecting the deviation of $\tau$ from its mean.  On the other hand, we can 
write a given mean curvature $\tau$ as
\bel{near-cmc-3}
\tau = \Ta +\Tbeta
\ee
where $\Ta$ is constant and $\int_{S^1} \Tbeta = 0$. If $\abs{\Ta}$ is large relative to,
say, $\frac{1}{2\pi}\int_{S^1}\abs\Tbeta$, then the ratio \eqref{near-cmc-2}
will be small (and $\tau$ will be near-CMC).  This weaker notion of being near-CMC
is similar to one used in \cite{isenberg-omurchadha-noncmc}.
It is not violated by the mean curvatures of Proposition \ref{prop:doubling},
and extends to the rough mean curvatures considered in  Section \ref{sec:roughfamily}.
\end{document}